\documentclass{article} 
    \usepackage{graphicx} 
    \usepackage{adjustbox} 
    \usepackage{color} 
    \usepackage{enumerate} 
    \usepackage{geometry} 
    \usepackage{amsmath} 
    \usepackage{amssymb} 
    \usepackage[mathletters]{ucs} 
    \usepackage[utf8x]{inputenc} 
    \usepackage{fancyvrb} 
    \usepackage{grffile} 
    \usepackage{hyperref}
    \usepackage{longtable} 
    \usepackage{booktabs}  
    
	  \usepackage{zanadu}
    \definecolor{orange}{cmyk}{0,0.4,0.8,0.2}
    \definecolor{darkorange}{rgb}{.71,0.21,0.01}
    \definecolor{darkgreen}{rgb}{.12,.54,.11}
    \definecolor{myteal}{rgb}{.26, .44, .56}
    \definecolor{gray}{gray}{0.45}
    \definecolor{lightgray}{gray}{.95}
    \definecolor{mediumgray}{gray}{.8}
    \definecolor{inputbackground}{rgb}{.95, .95, .85}
    \definecolor{outputbackground}{rgb}{.95, .95, .95}
    \definecolor{traceback}{rgb}{1, .95, .95}
    \definecolor{red}{rgb}{.6,0,0}
    \definecolor{green}{rgb}{0,.65,0}
    \definecolor{brown}{rgb}{0.6,0.6,0}
    \definecolor{blue}{rgb}{0,.145,.698}
    \definecolor{purple}{rgb}{.698,.145,.698}
    \definecolor{cyan}{rgb}{0,.698,.698}
    \definecolor{lightgray}{gray}{0.5}
    
    \definecolor{darkgray}{gray}{0.25}
    \definecolor{lightred}{rgb}{1.0,0.39,0.28}
    \definecolor{lightgreen}{rgb}{0.48,0.99,0.0}
    \definecolor{lightblue}{rgb}{0.53,0.81,0.92}
    \definecolor{lightpurple}{rgb}{0.87,0.63,0.87}
    \definecolor{lightcyan}{rgb}{0.5,1.0,0.83}
    
    \DefineVerbatimEnvironment{Highlighting}{Verbatim}{commandchars=\\\{\}}

\makeatletter
\def\PY@reset{\let\PY@it=\relax \let\PY@bf=\relax%
    \let\PY@ul=\relax \let\PY@tc=\relax%
    \let\PY@bc=\relax \let\PY@ff=\relax}
\def\PY@tok#1{\csname PY@tok@#1\endcsname}
\def\PY@toks#1+{\ifx\relax#1\empty\else%
    \PY@tok{#1}\expandafter\PY@toks\fi}
\def\PY@do#1{\PY@bc{\PY@tc{\PY@ul{%
    \PY@it{\PY@bf{\PY@ff{#1}}}}}}}
\def\PY#1#2{\PY@reset\PY@toks#1+\relax+\PY@do{#2}}

\expandafter\def\csname PY@tok@gd\endcsname{\def\PY@tc##1{\textcolor[rgb]{0.63,0.00,0.00}{##1}}}
\expandafter\def\csname PY@tok@gu\endcsname{\let\PY@bf=\textbf\def\PY@tc##1{\textcolor[rgb]{0.50,0.00,0.50}{##1}}}
\expandafter\def\csname PY@tok@gt\endcsname{\def\PY@tc##1{\textcolor[rgb]{0.00,0.27,0.87}{##1}}}
\expandafter\def\csname PY@tok@gs\endcsname{\let\PY@bf=\textbf}
\expandafter\def\csname PY@tok@gr\endcsname{\def\PY@tc##1{\textcolor[rgb]{1.00,0.00,0.00}{##1}}}
\expandafter\def\csname PY@tok@cm\endcsname{\let\PY@it=\textit\def\PY@tc##1{\textcolor[rgb]{0.25,0.50,0.50}{##1}}}
\expandafter\def\csname PY@tok@vg\endcsname{\def\PY@tc##1{\textcolor[rgb]{0.10,0.09,0.49}{##1}}}
\expandafter\def\csname PY@tok@m\endcsname{\def\PY@tc##1{\textcolor[rgb]{0.40,0.40,0.40}{##1}}}
\expandafter\def\csname PY@tok@mh\endcsname{\def\PY@tc##1{\textcolor[rgb]{0.40,0.40,0.40}{##1}}}
\expandafter\def\csname PY@tok@go\endcsname{\def\PY@tc##1{\textcolor[rgb]{0.53,0.53,0.53}{##1}}}
\expandafter\def\csname PY@tok@ge\endcsname{\let\PY@it=\textit}
\expandafter\def\csname PY@tok@vc\endcsname{\def\PY@tc##1{\textcolor[rgb]{0.10,0.09,0.49}{##1}}}
\expandafter\def\csname PY@tok@il\endcsname{\def\PY@tc##1{\textcolor[rgb]{0.40,0.40,0.40}{##1}}}
\expandafter\def\csname PY@tok@cs\endcsname{\let\PY@it=\textit\def\PY@tc##1{\textcolor[rgb]{0.25,0.50,0.50}{##1}}}
\expandafter\def\csname PY@tok@cp\endcsname{\def\PY@tc##1{\textcolor[rgb]{0.74,0.48,0.00}{##1}}}
\expandafter\def\csname PY@tok@gi\endcsname{\def\PY@tc##1{\textcolor[rgb]{0.00,0.63,0.00}{##1}}}
\expandafter\def\csname PY@tok@gh\endcsname{\let\PY@bf=\textbf\def\PY@tc##1{\textcolor[rgb]{0.00,0.00,0.50}{##1}}}
\expandafter\def\csname PY@tok@ni\endcsname{\let\PY@bf=\textbf\def\PY@tc##1{\textcolor[rgb]{0.60,0.60,0.60}{##1}}}
\expandafter\def\csname PY@tok@nl\endcsname{\def\PY@tc##1{\textcolor[rgb]{0.63,0.63,0.00}{##1}}}
\expandafter\def\csname PY@tok@nn\endcsname{\let\PY@bf=\textbf\def\PY@tc##1{\textcolor[rgb]{0.00,0.00,1.00}{##1}}}
\expandafter\def\csname PY@tok@no\endcsname{\def\PY@tc##1{\textcolor[rgb]{0.53,0.00,0.00}{##1}}}
\expandafter\def\csname PY@tok@na\endcsname{\def\PY@tc##1{\textcolor[rgb]{0.49,0.56,0.16}{##1}}}
\expandafter\def\csname PY@tok@nb\endcsname{\def\PY@tc##1{\textcolor[rgb]{0.00,0.50,0.00}{##1}}}
\expandafter\def\csname PY@tok@nc\endcsname{\let\PY@bf=\textbf\def\PY@tc##1{\textcolor[rgb]{0.00,0.00,1.00}{##1}}}
\expandafter\def\csname PY@tok@nd\endcsname{\def\PY@tc##1{\textcolor[rgb]{0.67,0.13,1.00}{##1}}}
\expandafter\def\csname PY@tok@ne\endcsname{\let\PY@bf=\textbf\def\PY@tc##1{\textcolor[rgb]{0.82,0.25,0.23}{##1}}}
\expandafter\def\csname PY@tok@nf\endcsname{\def\PY@tc##1{\textcolor[rgb]{0.00,0.00,1.00}{##1}}}
\expandafter\def\csname PY@tok@si\endcsname{\let\PY@bf=\textbf\def\PY@tc##1{\textcolor[rgb]{0.73,0.40,0.53}{##1}}}
\expandafter\def\csname PY@tok@s2\endcsname{\def\PY@tc##1{\textcolor[rgb]{0.73,0.13,0.13}{##1}}}
\expandafter\def\csname PY@tok@vi\endcsname{\def\PY@tc##1{\textcolor[rgb]{0.10,0.09,0.49}{##1}}}
\expandafter\def\csname PY@tok@nt\endcsname{\let\PY@bf=\textbf\def\PY@tc##1{\textcolor[rgb]{0.00,0.50,0.00}{##1}}}
\expandafter\def\csname PY@tok@nv\endcsname{\def\PY@tc##1{\textcolor[rgb]{0.10,0.09,0.49}{##1}}}
\expandafter\def\csname PY@tok@s1\endcsname{\def\PY@tc##1{\textcolor[rgb]{0.73,0.13,0.13}{##1}}}
\expandafter\def\csname PY@tok@kd\endcsname{\let\PY@bf=\textbf\def\PY@tc##1{\textcolor[rgb]{0.00,0.50,0.00}{##1}}}
\expandafter\def\csname PY@tok@sh\endcsname{\def\PY@tc##1{\textcolor[rgb]{0.73,0.13,0.13}{##1}}}
\expandafter\def\csname PY@tok@sc\endcsname{\def\PY@tc##1{\textcolor[rgb]{0.73,0.13,0.13}{##1}}}
\expandafter\def\csname PY@tok@sx\endcsname{\def\PY@tc##1{\textcolor[rgb]{0.00,0.50,0.00}{##1}}}
\expandafter\def\csname PY@tok@bp\endcsname{\def\PY@tc##1{\textcolor[rgb]{0.00,0.50,0.00}{##1}}}
\expandafter\def\csname PY@tok@c1\endcsname{\let\PY@it=\textit\def\PY@tc##1{\textcolor[rgb]{0.25,0.50,0.50}{##1}}}
\expandafter\def\csname PY@tok@kc\endcsname{\let\PY@bf=\textbf\def\PY@tc##1{\textcolor[rgb]{0.00,0.50,0.00}{##1}}}
\expandafter\def\csname PY@tok@c\endcsname{\let\PY@it=\textit\def\PY@tc##1{\textcolor[rgb]{0.25,0.50,0.50}{##1}}}
\expandafter\def\csname PY@tok@mf\endcsname{\def\PY@tc##1{\textcolor[rgb]{0.40,0.40,0.40}{##1}}}
\expandafter\def\csname PY@tok@err\endcsname{\def\PY@bc##1{\setlength{\fboxsep}{0pt}\fcolorbox[rgb]{1.00,0.00,0.00}{1,1,1}{\strut ##1}}}
\expandafter\def\csname PY@tok@mb\endcsname{\def\PY@tc##1{\textcolor[rgb]{0.40,0.40,0.40}{##1}}}
\expandafter\def\csname PY@tok@ss\endcsname{\def\PY@tc##1{\textcolor[rgb]{0.10,0.09,0.49}{##1}}}
\expandafter\def\csname PY@tok@sr\endcsname{\def\PY@tc##1{\textcolor[rgb]{0.73,0.40,0.53}{##1}}}
\expandafter\def\csname PY@tok@mo\endcsname{\def\PY@tc##1{\textcolor[rgb]{0.40,0.40,0.40}{##1}}}
\expandafter\def\csname PY@tok@kn\endcsname{\let\PY@bf=\textbf\def\PY@tc##1{\textcolor[rgb]{0.00,0.50,0.00}{##1}}}
\expandafter\def\csname PY@tok@mi\endcsname{\def\PY@tc##1{\textcolor[rgb]{0.40,0.40,0.40}{##1}}}
\expandafter\def\csname PY@tok@gp\endcsname{\let\PY@bf=\textbf\def\PY@tc##1{\textcolor[rgb]{0.00,0.00,0.50}{##1}}}
\expandafter\def\csname PY@tok@o\endcsname{\def\PY@tc##1{\textcolor[rgb]{0.40,0.40,0.40}{##1}}}
\expandafter\def\csname PY@tok@kr\endcsname{\let\PY@bf=\textbf\def\PY@tc##1{\textcolor[rgb]{0.00,0.50,0.00}{##1}}}
\expandafter\def\csname PY@tok@s\endcsname{\def\PY@tc##1{\textcolor[rgb]{0.73,0.13,0.13}{##1}}}
\expandafter\def\csname PY@tok@kp\endcsname{\def\PY@tc##1{\textcolor[rgb]{0.00,0.50,0.00}{##1}}}
\expandafter\def\csname PY@tok@w\endcsname{\def\PY@tc##1{\textcolor[rgb]{0.73,0.73,0.73}{##1}}}
\expandafter\def\csname PY@tok@kt\endcsname{\def\PY@tc##1{\textcolor[rgb]{0.69,0.00,0.25}{##1}}}
\expandafter\def\csname PY@tok@ow\endcsname{\let\PY@bf=\textbf\def\PY@tc##1{\textcolor[rgb]{0.67,0.13,1.00}{##1}}}
\expandafter\def\csname PY@tok@sb\endcsname{\def\PY@tc##1{\textcolor[rgb]{0.73,0.13,0.13}{##1}}}
\expandafter\def\csname PY@tok@k\endcsname{\let\PY@bf=\textbf\def\PY@tc##1{\textcolor[rgb]{0.00,0.50,0.00}{##1}}}
\expandafter\def\csname PY@tok@se\endcsname{\let\PY@bf=\textbf\def\PY@tc##1{\textcolor[rgb]{0.73,0.40,0.13}{##1}}}
\expandafter\def\csname PY@tok@sd\endcsname{\let\PY@it=\textit\def\PY@tc##1{\textcolor[rgb]{0.73,0.13,0.13}{##1}}}

\makeatother

\newcommand{\be}{\begin{equation}}
\newcommand{\ee}{\end{equation}}
\newcommand{\Prob}{\mathbb{P}}
\newcommand{\R}{\mathbb{R}}
\newcommand{\esp}{\mathbb{E}}

\newcommand{\mino}{<}
\newcommand{\maj}{>}

\newcommand{\til}{~}
\renewcommand{\Re}{\mathrm{Re}}
\renewcommand{\Im}{\mathrm{Im}}

    \definecolor{incolor}{rgb}{0.0, 0.0, 0.5}
    \definecolor{outcolor}{rgb}{0.545, 0.0, 0.0}
    
    \hypersetup{
      breaklinks=true,  
      colorlinks=true,
      urlcolor=blue,
      linkcolor=darkorange,
      citecolor=darkgreen,
      }
    
\begin{document}
    
     \hypersetup{ breaklinks=true,  
      colorlinks=true,
      urlcolor=blue,
      linkcolor=darkorange,
      citecolor=darkgreen,      
      }

\newpage 

\section{Introduction}

We consider an asset price $S_T$ at some fixed date $T > 0$.
We denote $\mathbb P$ the pricing (T-forward) measure, so that the price of a call option with maturity $T$ and strike $K$ is given by $B(0,T) \esp^{\Prob}[(S_T - K)^+]$, where $B(0,T)$ denotes the current price of the zero-coupon bond with maturity $T$.
We work with dimensionless quantities: we denote $k = \log(K/F)$ the forward log-strike, where $F = \esp^{\Prob}[S_T]$ denotes the forward price for maturity $T$, and $v(k) = \sqrt T \sigma_{\mathrm{BS}}(T,k)$ the dimensionless (or ``total'')
 implied volatility.
Recall that $v(k)$ is defined for all
$k \in \mathbb R$ by the equation $\mathrm{Call_{BS}}(k,v(k)) = \esp \bigl[\bigl(\frac{S_T}F - e^k\bigl)^+ \bigr]$, where $\mathrm{Call_{BS}}(k,v) = N(d_1(k,v)) - e^k N(d_2(k,v))$, $N(\cdot)$ is the standard normal cdf, and $d_i(k,v) = \frac{-k}{v} + (\frac 32 -i) v$.

It is well-known that any $C^2$ (or convex) payoff $\varphi(S_T)$ with linear growth can be statically replicated with a strip of call and put options, so that its price can be written via Carr--Madan's formula $B(0,T)\esp[\varphi(S_T)] = B(0,T)\varphi(F) + \int_0^F \varphi''(K) P(K) dK + \int_F^{\infty} \varphi''(K) C(K) dK$, where $P(K)$ and $C(K)$ denote put and call prices for the maturity $T$, see e.g.\til\cite[Eq (11.1)]{GathBook}.
On the other hand, if the law of $S_T$ under $\Prob$ is absolutely continuous 
with respect to the Lebesgue measure on $[0,\infty)$, using the well-known Breeden-Litzenberger relations,
one can write
\be \label{e:density} 
\esp[\varphi(S_T)] = \int_0^\infty \varphi(K) \frac{d^2}{dK^2} \esp[(S_T - K)^+] dK = 
F \int_\R \varphi(K) \frac{d^2}{dK^2} \mathrm{Call_{BS}}(k,v(k))|_{k = \log(K/F)} \, dK.
\ee
Applying the chain rule to the rightmost integrand in \eqref{e:density} leads to an integral formula containing Black-Scholes Greeks with respect to strike and volatility, and the derivatives of the implied volatility smile $v(\cdot)$ up to order two.
A stream of literature \cite{Mat2000, GathBook,Fukasawa2012,Bergomi2016} studies the possibility of re-expressing Equation \eqref{e:density} in such a way that the derivatives of the implied volatility do not appear any more on the right hand side. This is a relevant feature in practice, because observed market data is (in any case) discrete.
Such investigations required the introduction of the concept of normalizing transformation of the implied volatility smile, introduced by Chriss and Morokoff \cite{{ChrissMorok99}} and Matytsin \cite{Mat2000} and formalized in the seminal work of Fukasawa \cite{Fukasawa2012}, that we recall below.

One of the most important examples in this field is the following formula for the implied variance of the log contract\footnote{which coincides with the fair strike of the Variance Swap, under the assumption that $(S_t)_{t \le T}$ follows a diffusion process.
When $T=30\;\mathrm{days}$ and $S$ is the S\&P500 stock index, the left hand side of \eqref{e:varSwapFormula} defines the (theoretical value) of $\mbox{VIX}^2$ at $t=0$.
}, see Chriss and Morokoff \cite{ChrissMorok99} or Gatheral \cite{GathBook}: 
\be \label{e:varSwapFormula}
\esp^{\Prob} \left[-\frac 2T \log \biggl(\frac{S_T}{F} \biggr) \right]
= \frac 1T \int_{\R} v(g_2(z))^2 \phi(z) dz.
\ee
In \eqref{e:varSwapFormula}, $\phi$ is the standard normal density, and $g_2: \R \to \R$ is the inverse of the function (called \emph{second normalizing transformation})
\[
f_2(k) := -d_2(k,v(k)) = \frac{k}{v(k)} + \frac{v(k)}2.
\]
Similarly, the first normalizing transformation (used later on) is given by $f_1(k) := -d_1(k,v(k)) = \frac{k}{v(k)} - \frac{v(k)}2$.
Apart from its appealing compactness, the formula \eqref{e:varSwapFormula} is amenable for numerical approximations, notably in view of the use of Gauss-Hermite quadrature -- see the discussions in \cite[Remark 4.9]{Fukasawa2012} and Bergomi \cite[p.\til143]{Bergomi2016}.
Other examples of similar formulas include: other derivatives such as the $S \ln S$ contract (related to the Gamma Swap, see again \cite{Fukasawa2012} and Section \ref{s:duality} of this paper), and a formula for the characteristic function of $X_T = \log(S_T/F)$ due to Matytsin \cite{Mat2000}, see below. 
The important property that the map $f_2: \R \to \R$ is actually invertible for \emph{any} arbitrage-free implied volatility $v(\cdot)$, implicitly assumed in the aforementioned works, was first proven by Fukasawa \cite{Fukasawa2012}.

\textbf{Matytsin's formula for characteristic functions and
Bergomi's formula for $\mathbb{E}[ (S_T/F)^p ]$}. 

Denote $v_2(z) = v(g_2(z))$. 
Assuming that $v_2$ is differentiable, Matytsin \cite{Mat2000} gives the following formula for the characteristic function of $X_T$
\be \label{e:MatytsinFormula}
\mathbb E \left[e^{i \eta X_T} \right] = \int_{\mathbb R} e^{-i \eta v_2(z) \left(\frac12 v_2(z) - z \right) } \left(1 - i \eta v_2'(z) \right) \phi(z) dz, 
\qquad \eta \in \mathbb R.
\ee
The proof of \eqref{e:MatytsinFormula} is (only) sketched in Matytsin's slides.
Building on this work, Bergomi \cite[Section 4.3.1]{Bergomi2016} derives a formula for the moments of $S_T/F$ of order
$p \in [0,1]$:
\begin{equation} \label{BergomiFormula}
\mathbb{E}\left[\left(\frac{S_T}{F}\right)^p \right] = \int_{\mathbb R} e^{\frac12 p(p-1) v^{p}(z)^2} \phi(z) dz,
\qquad p \in [0,1],
\end{equation}
where the ``$p$-normalized'' implied
volatility $v^{p}(\cdot)$ is defined in the following way: consider the
convex interpolation
\[
f(p,k) = p f_1(k) + (1-p) f_2(k)
 = \frac{k}{v(k)} + \left(\frac12-p\right) \frac{v(k)}2
\]
of the two normalizing transformations $f_1$ and
$f_2$.
We know from Fukasawa \cite{Fukasawa2012} that the two maps $k \mapsto f_1(k)$ and $k \mapsto f_2(k)$ are strictly
increasing from $\mathbb R$ to $\mathbb R$: therefore, so is
$k \mapsto f(p,k)$, for every $p \in [0,1]$. Let bow $g(p,\cdot)$ be the
inverse of $f(p,\cdot)$ on $\mathbb R$: $v^{p}(\cdot)$ is defined by
\be \label{e:pNormalizedImplVolDef}
v^{p}(z) = v(g(p,z)), 
\qquad
\mbox{for all } z \in \mathbb R
\ee
(hence, with reference to Fukasawa's notation, we have $v^{1} = v_1$, $v^{0} = v_2$).
Note we have the following nice interpretation of \eqref{BergomiFormula}: in the Black-Scholes model, where $S_T = F e^{\sigma W_T - \frac 12 \sigma^2 T}$ is a geometric Brownian motion with constant volatility parameter $\sigma = \frac v{\sqrt T}$, one has $\esp\bigl[\bigl(\frac{S_T}{F}\bigr)^p \bigr] = e^{\frac12 p(p-1)v^2} = \int_{\R} e^{\frac12 p(p-1) v^2} \phi(z) dz$.
Therefore, we can see Equation \eqref{BergomiFormula} as an extension of the pricing formula for power payoffs, from the Black-Scholes world to models with non-constant implied volatility.

The formulas \eqref{e:MatytsinFormula} and \eqref{BergomiFormula} are the starting point of this work.
As mentioned above, Bergomi \cite{Bergomi2016} derives \eqref{BergomiFormula}
from \eqref{e:MatytsinFormula}.
Here, we will follow a different route: our starting point is the work of Fukasawa. We first extend the formula for expectations of functions of $X_T$ with polynomial growth given in \cite[Theorem 4.6]{Fukasawa2012} to exponential functions --  carrying out in details the plan addressed in \cite[Remark 4.8]{Fukasawa2012}.
This provides a formula for the generalized characteristic function
$p \in \mathbb{C} \mapsto \mathbb E [e^{p X_T}]$ on its analyticity domain, written directly in terms of the implied volatility smile.
Matytsin's \eqref{e:MatytsinFormula} and Bergomi's \eqref{BergomiFormula} formulas are embedded in this representation as special cases (along with a dual version of the first, and an extension to the complex plane of the second).
By taking real values of $p$, this formula allows to numerically evaluate the (finite) risk-neutral moments of the underlying asset price from the market smile -- therefore identifying model-free quantities that can be used as targets in the calibration of a parametric model.

As addressed in \cite[Remark 4.8]{Fukasawa2012}, it is natural, when evaluating expectations of the form $\esp\bigl[\bigl(\frac{S_T}{F}\bigr)^p \bigr]$, to exploit Lee's moment formulas \cite{Lee04} relating the critical moments of $S_T$ to the asymptotic slopes of the implied volatility for large and small strikes. 
We stress that our approach here goes the other way round: we prove an integral representation for $\esp\bigl[\bigl(\frac{S_T}{F}\bigr)^p \bigr]$ \emph{without} making use of Lee's result.
Then, as a by-product, we can deduce sharp bounds on the exponents that appear in the moment formulas.

\textbf{The example of the SSVI parameterisation}.
Gatheral and Jacquier \cite{SSVI} propose the following parameterisation for total implied variance (the square of the dimensionless implied volatility $v$):

\begin{equation} \label{SSVI}
v_{\mathrm{SSVI}}^2(k) = \frac{\theta_T}2 \left(1 + \rho \varphi(\theta_T) k + \sqrt{(\varphi(\theta_T) k + \rho)^2 + 1 -\rho^2} \right),
\end{equation}

where $\theta_T > 0$, $\varphi: (0,\infty) \to \mathbb R_+ $ and $\rho \in (-1,1)$. The RHS of \eqref{SSVI} provides a parameterisation of the whole implied variance surface, for every log-strike $k$ and
every maturity $T$. In the present setting, we are interested in parameterisations of a single arbitrage-free smile for a fixed maturity $T$: for simplicity, we will therefore drop the index $T$ from the notation, and denote $\theta = \theta_T$ and $\varphi = \varphi(\theta)$.
Important no-arbitrage properties of the SSVI model will be recalled in Section \ref{the-ssvi-parameterisation}.

\section{Extension of Fukasawa's pricing formula} \label{extension-of-thm-4.6-in-fukasawa-mf2012}

Our standing assumptions on the implied volatility are the following:

\begin{assumption} \label{standingAssumptionsV}
\begin{enumerate}
\def\labelenumi{(\roman{enumi})}
\item
  $v$ is twice differentiable on $\mathbb R$ and $v(k) > 0$ for all
  $k \in \mathbb{R}$.
\item
  $\lim_{k \to -\infty} d_2(k,v(k)) =\lim_{k \to -\infty} \left( \frac{-k}{v(k)} - \frac{v(k)}2 \right) = +\infty$
\end{enumerate}
\end{assumption}

Denote $\mu$ the law of $S_T/F$ under $\Prob$.
It is classical that the second differentiability of $v(\cdot)$ is equivalent to the existence of
a density with respect to Lebesgue measure for the restriction of $\mu$ to $(0,\infty)$.
Moreover, the strict positivity of $v(\cdot)$ is equivalent to the two conditions $\inf \{ \mbox{supp} (\mu)\} = 0$ and $\sup \{ \mbox{supp} (\mu) \} = \infty$, see \cite{RogTeh}.
In its turn, Assumption \ref{standingAssumptionsV} (ii) is equivalent to
$\mu(\{0\}) = \mathbb P(S_T = 0) = 0$. (In general, we have
$\lim_{k \to -\infty} d_2(k,v(k)) = -N^{-1}(\mathbb P(S_T = 0))$, see \cite{TehPresentation}, \cite{Fukasawa2010}.)

\begin{remark}
Recall that $f_2(k) \ge \sqrt{2k}$ for every $k > 0$ from the arithmetic mean-geometric mean inequality, therefore we always have $\lim_{k \to \infty} f_2(k) = \infty$.
Analogously, $f_1(k) = \frac{-|k|}{v(k)} - \frac{v(k)}2 \le -\sqrt{2|k|}$ for every $k < 0$, hence $\lim_{k \to -\infty} f_1(k) = -\infty$.
The limit of $f_1$ as $k \to +\infty$ is related to the arbitrage freeness of $v(k)$: indeed, $\lim_{k \to +\infty} f_1(k) = +\infty$ is equivalent to the no-arbitrage condition $\lim_{k \to +\infty} \mathrm{Call_{BS}}(k,v(k)) = 0$.
Therefore. $f_1$ always maps $\R$ onto $\R$.
Assumption \ref{standingAssumptionsV} (ii) ensures that we have $\lim_{k \to -\infty} f_2(k) = -\infty$, so that $f_2$ is surjective, too.
\end{remark}

\hspace{4mm} Fukasawa \cite{Fukasawa2012} proves the following result:

\begin{theorem}[Theorem 4.6 in \cite{Fukasawa2012}] \label{FukasawaFormulaThm} 
Let $\Psi$ be an absolutely continuous function with derivative $\Psi'$ of polynomial growth, and assume that there exists $q>0$ such that $\mathbb E[S_T^{-q}] < \infty$. Then,
\begin{equation} \label{PsiFormula}
\mathbb E \left[\Psi\left( \log \frac{S_T}{F} \right) \right] 
= \int_{-\infty}^{+\infty} \bigl[\Psi(g_2(z)) - \Psi'(g_2(z)) + \Psi'(g_1(z)) e^{-g_1(z)} \bigr] \phi(z) dz.
\end{equation}
Analogously, if there exists $q>0$ such that $\mathbb E[S_T^{1+q}] < \infty$, then $\mathbb E \left[\frac{S_T}{F} \Psi\left( \log \frac{S_T}{F} \right) \right] 
= \int_{-\infty}^{+\infty} \bigl[\Psi(g_1(z)) + \Psi'(g_1(z)) - \Psi'(g_2(z)) e^{g_2(z)} \bigr] \phi(z) dz$.
\end{theorem}

Equation \eqref{PsiFormula} can be proven starting from \eqref{e:density} and then applying judicious integration by parts, and change of variables by means of the transformations $f_1$ and $f_2$.
Important steps in this proof, see \cite{Fukasawa2012}, are (i) ensuring the integrability of the involved integrands, and (ii) checking that the boundary terms arising from integration by parts do give zero contribution (point (ii) amounts to showing that $\lim_{k \to \pm \infty} \Psi(k) v'(k) \phi(f_2(k))=0$, see Lemma \ref{auxiliaryIntegrabilityAndLimits} in the Appendix).

Setting $\Psi(k) = e^{pk}$ in \eqref{PsiFormula} would give us a formula
for the moment of $S_T/F$ of order $p$. Unless $p=0$, such a function $\Psi$ falls outside the class of functions covered by Theorem
\ref{FukasawaFormulaThm}.
We therefore proceed to extend Theorem \ref{FukasawaFormulaThm} to this setting.

\subsection{Functions with exponential
growth}\label{functions-with-exponential-growth}

Define the two functions
\be \label{e:pPlusMinusDef}
p_+(\beta) = \frac12 \left(\frac 1{\beta} + \frac {\beta}4 + 1 \right),
\qquad
p_-(\beta) = \frac12 \left(\frac 1{\beta} + \frac {\beta}4 - 1 \right),
\qquad \beta \in (0,2],
\ee
and set $p_\pm(\beta) = +\infty$ if $\beta = 0$. It is easy to see that $p_+(\beta) \ge 1$ and $p_-(\beta) \ge 0$, for all $\beta \in [0,2]$.
The expressions above of the functions $p_\pm(\cdot)$ arise from certain integrability conditions of the integrand in \eqref{PsiFormula} when $\Psi(k) = e^{pk}$, as we will explain more precisely below. 
Now, if we denote
\be \label{e:criticalExpDef} 
p^* = \sup \{p > 0: \esp[(S_T)^p] < \infty \},
\qquad
q^* = \sup \{q > 0: \esp[(S_T)^{-q}] < \infty \},
\ee
the right, resp. left, critical exponent of $S_T$ (note we define $q^*$ to be positive), then Roger Lee's moment formula \cite{Lee04} states that:
\be \label{e:momFormula}
p^* = p_+(\beta_+), \qquad q^* = p_-(\beta_-)
\ee
where
\be \label{e:wingSlopes}
\sqrt{\beta_+} = \limsup_{k \to \infty} \frac{v(k)}{\sqrt{k}},
\qquad 	\qquad 
\sqrt{\beta_-} = \limsup_{k \to -\infty} \frac{v(k)}{\sqrt{|k|}}.
\ee

As pointed out in the introduction, in the present work we do \emph{not} make use of Lee's result in order to extend Equation \eqref{PsiFormula} to functions $\Psi$ with exponential growth.
In other words, we do not use the information $p^* = p_+(\beta_+)$ or $q^* = p_-(\beta_-)$ at any point.

\begin{example} \label{ex:SSVI_slopes}
For the SSVI parameterisation \eqref{SSVI} we have
\be \label{e:SSVI_slopes}
\beta_\pm(\mathrm{SSVI}) = \theta \varphi \frac{(1\pm\rho)}{2}.
\ee 
As we will recall in section \ref{no-arbitrage-conditions} below, a necessary condition for no arbitrage is
$\theta \varphi (1\pm\rho) \leq 4$, so that $0 \leq \beta_\pm(\mathrm{SSVI}) \leq 2$.
\end{example}

\hspace{4mm} The quantitative link between \eqref{PsiFormula} and \eqref{e:momFormula} can be highlighted with some simple calculations: assume (only within this paragraph) that the first equation in \eqref{e:wingSlopes} hold as a limit, and that $\beta_+ \notin \{0, 2 \}$. Then, from the definition of $f_{1}$ and $f_{2}$,
\be \label{e:f12Asymptotics}
f_2(k) \sim \left(\frac 1{\sqrt{\beta_+}} + \frac{\sqrt{\beta_+}}2 \right) \sqrt{k},
\qquad \qquad
f_1(k) \sim \left(\frac 1{\sqrt{\beta_+}} - \frac{\sqrt{\beta_+}}2 \right) \sqrt{k}
\qquad \qquad
\mbox{as } k \to \infty.
\ee
Using the definition $f_i(g_i(z)) = z$ of the maps $g_{i}$, it is easy to see that Eq \eqref{e:f12Asymptotics} implies a quadratic behavior of $g_1$ and $g_2$ for large $z$:
\be \label{e:g12Asymptotics}
g_2(z)
\sim \frac{z^2}{\Bigl(\frac 1{\sqrt{\beta_+}} + \frac{\sqrt{\beta_+}}2\Bigr)^2}
= \frac{z^2}{2 p_+(\beta_+)},
\qquad \quad
g_1(z)
\sim \frac{z^2}{\Bigl(\frac 1{\sqrt{\beta_+}} - \frac{\sqrt{\beta_+}}2\Bigr)^2}
= \frac{z^2}{2 p_-(\beta_+)}
\qquad \quad
\mbox{as } z \to \infty.
\ee
Now, a formal application of \eqref{PsiFormula} to $\Psi(k) = e^{pk}$ leads to $\esp\bigl[\bigl(\frac{S_T}{F}\bigr)^p \bigr] = \int_{\R} \bigl[p e^{(p-1) g_1(z)} + (1-p) e^{p g_2(z)} \bigr] \phi(z) dz$.
Applying \eqref{e:g12Asymptotics}, a straightforward calculation yields
\[
\begin{array}{ll}
l_1(z) := e^{(p-1) g_1(z)} \phi(z)
= \exp \Bigl(\frac12 \frac{p - p_+(\beta_+)}{p_-(\beta_+)} z^2 \ (1+o(1)) \Bigr)
&\qquad \mbox{as } z \to \infty,
\\ \\
l_2(z) := e^{p g_2(z)} \phi(z)
= \exp \Bigl(\frac12 \frac{p - p_+(\beta_+)}{p_+(\beta_+)} z^2 \ (1+o(1)) \Bigr)
&\qquad \mbox{as } z \to \infty.
\end{array}
\]
The last estimates above show that $p = p_+(\beta_+)$ is precisely the threshold between Gaussian integrability (when $p \mino p_+(\beta_+)$) or, on the contrary, exploding behavior (when $p \maj p_+(\beta_+)$) for the two functions $l_1$ and $l_2$ inside the integral (for positive $z$).
The analogous argument allows to show that  $p = -p_-(\beta_-)$ is the threshold between integrability and explosion of the integrand for negative $z$.
Some more work (which we perform in the proof of Theorem \ref{extensionFukThm} below) is required to study the asymptotic behavior of the linear combination $p \, l_1(\cdot) + (1-p) l_2(\cdot)$, and to deal with the remaining cases $\beta_{\pm} \in \{0,2\}$.


\begin{definition} \label{def:expGrowth}
We say that a function $\Psi: \mathbb R \to \mathbb R$ has exponential
growth of order $p$ for some $p \in \mathbb R$ if the function
$k \mapsto e^{-pk} \Psi(k)$ is bounded.
\end{definition}


\begin{lemma} \label{auxiliaryIntegrability}
Let $p \in \mathbb R$ such that
\[
- p_-(\beta_-) < p < p_+(\beta_+)
\]
where $p_{\pm}(\cdot)$ are defined in \eqref{e:pPlusMinusDef}, and let $\Psi$ be an absolutely continuous function such that $\Psi$ and $\Psi'$ have exponential growth of order $p$.
Then, the functions

\begin{equation} \label{integrability}
z \mapsto \Psi(g_2(z)) \phi(z),
\qquad z \mapsto \Psi'(g_2(z)) \phi(z),
\qquad \mbox{and} \qquad z \mapsto \Psi'(g_1(z)) e^{-g_1(z)} \phi(z)
\end{equation}

are integrable on $\mathbb R$.
\end{lemma}

With a view on Lemma \ref{auxiliaryIntegrability}, note that the bound $f_2(k) \ge \sqrt{2k}$ for $k \ge 0$ (resp.
$f_1(k) \le -\sqrt{2|k|}$ for $k \le 0$) entails

\begin{equation} \label{g2estimate}
g_2(z) \le \frac{z^2}2 \mbox{ for $z$ large enough} \qquad \mbox{(resp. $g_1(z) \ge -\frac{z^2}2$ for $z$ small enough).}
\end{equation}

Since $g_2(z)$ is eventually positive for large $z$, the estimate \eqref{g2estimate} follows from
$z = f_2(g_2(z)) \ge \sqrt{2 g_2(z)}$ (resp.
$z = f_1(g_1(z)) \le -\sqrt{2 |g_1(z)|}$ for $z$ sufficiently small).
Therefore,  for every $p \in (0,1)$,
\begin{equation} \label{integrabilityForPIn01}
\begin{array}{lll}
p g_2(z) \to_{z \to -\infty} -\infty, & \qquad & e^{p g_2(z)} \phi(z) \le e^{(p-1) \frac{z^2}2} \ \mbox{ for $z$ large enough}
\\
\\
e^{(p-1) g_1(z)} \phi(z) \le e^{-p\frac{z^2}2} \ \mbox{ for $z$ small enough}, & \qquad & (p-1) g_1(z) \to_{z \to \infty} -\infty.
\end{array}
\end{equation}
For every function $F(\cdot)$ with exponential growth of order $p$, we have
$|F(g_2(z)) \phi(z)| \le \mbox{const.} \times e^{p g_2(z)} \phi(z)$ and
$|F(g_1(z)) e^{-g_1(z)} \phi(z)| \le \mbox{const.} \times e^{(p-1) g_1(z)} \phi(z)$.
Consequently, the estimates in \eqref{integrabilityForPIn01} show that the functions in
\eqref{integrability} are always integrable for every $p \in (0,1) = (p_-(2), p_+(2))$ and for every arbitrage-free smile $v$ (regardless of the values of $\beta_\pm$).

\begin{theorem} \label{extensionFukThm}
Let $p \in \mathbb R$ such that
\begin{equation} \label{pBounds}
- p_-(\beta_-) < p < p_+(\beta_+)
\end{equation}
where $p_{\pm}(\cdot)$ are defined in \eqref{e:pPlusMinusDef} and $\beta_\pm$ in \eqref{e:wingSlopes}. Let $\Psi$ be an absolutely continuous function such that $\Psi$ and $\Psi'$ have exponential growth of order $p$.
Then, 
\begin{equation} \label{extensionFukasawa}
\int_{-\infty}^{+\infty} \bigl[\Psi(g_2(z)) - \Psi'(g_2(z)) + \Psi'(g_1(z)) e^{-g_1(z)} \bigr] \phi(z) dz =
\mathbb E \left[\Psi\left( \log{\frac{S_T}{F}} \right) \right].
\end{equation}
In particular, for all $- p_-(\beta_-) < p < p_+(\beta_+)$,
\begin{equation} \label{moment}
L(p) := \int_{-\infty}^{+\infty} \bigl[p e^{(p-1) g_1(z)} + (1-p) e^{p g_2(z)} \bigr] \phi(z) dz =
\mathbb E \left[\left( \frac{S_T}{F} \right)^p \right] =: M(p) < \infty.
\end{equation}
\end{theorem}

The proofs of Lemma \ref{auxiliaryIntegrability} and Theorem \ref{extensionFukThm} are given in Appendix \ref{s:appendix}.
The formula \eqref{moment} for the moments of $S_T$ was mentioned in the introduction of \cite{Fukasawa2012}, the extension of Thm \ref{FukasawaFormulaThm} to functions with exponential growth being implicitly assumed therein.

\hspace{4mm} We stress once again that our proof of Theorem \ref{extensionFukThm} does not make use of Lee's result \cite{Lee04}.
As an immediate consequence of Theorem \ref{extensionFukThm}, we have the following bounds:

\begin{corollary} \label{c:LeeLowerBound}
Let $p_+(\beta_+), p_-(\beta_-)$ be defined by \eqref{e:pPlusMinusDef} and $p^*,q^*$ by \eqref{e:criticalExpDef}.
Then
\[
p^* \ge p_+(\beta_+)
\qquad \mbox{and} \qquad
q^* \ge p_-(\beta_-).
\]
In particular, $p^* = p_+(\beta_+)$ if $p_+(\beta_+) = \infty$ and $q^* = p_-(\beta_-)$ if $p_-(\beta_-) = \infty$.
\end{corollary}
\begin{proof}
From Equation \eqref{moment}, we have $M(p) \mino \infty$ for all $p \in (-p_-(\beta_-), p_+(\beta_+))$. 
The claim then follows from the definition of $p^*$ and $q^*$.
\end{proof}


\begin{remark}
Our proof of Theorem \ref{extensionFukThm} in Appendix \ref{s:appendix} is obtained essentially by rerunning the explicit computations linking the LHS to the RHS in \eqref{extensionFukasawa}, showing that i) all the required integrability conditions are met, and ii) that integration by parts produce zero boundary terms.
Another approach to the proof of \eqref{moment}
would go as follows: apply Fukasawa's result (Theorem \ref{FukasawaFormulaThm} above) to the functions $\Psi_n(x) = \left(1 + \frac{px}n \right)^n$, which have polynomial growth for every $n$.
Under the assumption $q^*\maj 0$, this yields $\esp[\Psi_n(X_T)] = \int \bigl[ \bigl(1 + \frac{p g_2(z)}n \bigr)^n  - p \bigl(1 + \frac{p g_2(z)}n \bigr)^{n-1} + p \bigl(1 + \frac{p g_1(z)}n \bigr)^{n-1} e^{-g_1(z)} \bigr]\phi(z) dz$.
By monotone convergence, the left hand side $\esp[\Psi_n(X_T)]$ converges to $\esp[e^{p X_T}]$ for every $p \in (-q^*,p^*)$. 
On the other hand, if $p \in (-p_-(\beta_-), p_+(\beta_+))$ the RHS converges to $\int \bigl[p e^{(p-1) g_1(z)} + (1-p) e^{p g_2(z)} \bigr] \phi(z) dz$ by dominated convergence, thanks to Lemma \ref{auxiliaryIntegrability}.
This argument proves Eq \eqref{moment} for every $p$ such that $\max(-p_-(\beta_-),-q^*) < p < \min(p_+(\beta_+), p^*)$; some additional work then allows to prove the inequalities $p^* \ge p_+(\beta_+)$ and $q^* \ge q_-(\beta_-)$ (we use a similar argument in Appendix \ref{s:appendix}), and the claimed result then follows for every $p$ in \eqref{pBounds}.

Following our (alternative) proof in Appendix \ref{s:appendix}, we get rid of the limitation $q^*\maj 0$ (or $p^* \maj 1$) from Theorem 4.6 in \cite{Fukasawa2012}.  
This is related to the growth condition we consider on the function $\Psi$, which is ``one-sided'': while, in Definition \ref{def:expGrowth}, $\Psi$ is allowed to grow faster than a polynome for large arguments (if, say, $p \maj 0$), on the other side $\Psi(x)$ has to go to zero for $x \to -\infty$ (as opposed to the polynomial growth condition in \cite[Theorem 4.6]{Fukasawa2012}, which allows $\Psi$ to diverge on both sides).
\end{remark}

\begin{remark}[About the proof of the converse inequalities for $p^*, p_+$
and $q^*, p_-$
, using Theorem \ref{extensionFukThm}]
The function $p \mapsto M(p) = \esp \bigl[ \bigl(\frac{S_T}F \bigr)^p \bigr]$ in Theorem \ref{extensionFukThm} is the bilateral Laplace transform of a probability measure on $\R$ (the law of $X_T$). 
Its positive and negative abscissas of convergence are, respectively, $p^*$ and $-q^*$ defined in \eqref{e:criticalExpDef}; moreover, $M$ has a unique holomorphic extension to the strip $D^* = \{p \in \mathbb{C}: -q^* < \mathrm{Re}(p) < p^*\}$.

Assume $p^* \maj p_+(\beta_+)$: then, $M$ would be a holomorphic extension of the function $p \mapsto L(p)$, defined on $(-p_-(\beta_-), p_+(\beta_+))$, to $D^*$. 
If we prove that this is not possible -- that is, that $L$ cannot be extended analytically to any neighbourhood of $p_+(\beta_+)$ -- then we would have shown by contradiction that $p^* \le p_+(\beta_+)$, therefore $p^* = p_+(\beta_+)$ by Lemma \ref{c:LeeLowerBound}. The symmetric argument of course applies to the inequality $q^* \le p_-(\beta_-)$.

It is not difficult to prove, using similar arguments to the proof of Theorem \ref{extensionFukThm}, that the function $L$ (as defined by the left hand side of Eq \eqref{moment}) is infinite for every $p \notin [-p_-(\beta_-), p_+(\beta_+)]$ if the implied volatility slope has a limit. More precisely, we can show:
\begin{itemize}
\item Assume that $\lim_{k \to \infty} \frac{v(k)}{\sqrt{k}}$ exists. Then, $L(p) = +\infty$ for every $p \maj p_+(\beta_+)$.
\item Assume that $\lim_{k \to -\infty} \frac{v(k)}{\sqrt{|k|}}$ exists. Then, $L(p) = +\infty$ for every $p \mino -p_-(\beta_-)$.
\end{itemize}
Note we are not (yet) showing here that $L$ does not admit a holomorphic extension to a neighbourhood of $p_+(\beta_+)$ or $-p_-(\beta_-)$.
Let us just observe, for the moment, that the proof of this statement (that we leave for future work) might not be so straightforward.
We can see $L$ as a linear combination of Laplace transforms of positive functions, $L(p) = p \, L_1(p) + (1-p) \, L_2(p)$. 
Since the coefficients $p$ and $1-p$ have \emph{opposite signs} as soon as $p \maj 1$ or $p \mino 0$, even if each Laplace transform $L_i$ cannot be extended analytically above $p_+(\beta_+)$ (or below $-p_-(\beta_-)$), the behavior of their linear combination is more subtle to study.
\end{remark}

\subsection{Extension to the complex plane}\label{extension-to-the-complex-plane}

We can extend Theorem \ref{extensionFukThm} to the following

\begin{proposition} \label{p:complexExtension}
The identity 
\begin{equation} \label{extensionComplexP}
\mathbb E \left[e^{p X_T} \right] 
= \int_{-\infty}^{+\infty} \bigl[ p e^{g_1(z)(p-1)} + (1-p) e^{p g_2(z)} \bigr] \phi(z) dz
\end{equation}
holds for all $p \in\mathbb C$ with $-p_-(\beta_-) < \mathrm{Re}(p) < p_+(\beta_+)$.
\end{proposition}

\begin{proof}
By Lemma \ref{auxiliaryIntegrability}, Equation \eqref{moment} defines a holomorphic function $L$ on $D = \{p \in \mathbb{C} : -p_-(\beta_-) < \mathrm{Re}(p) < p_+(\beta_+) \}$.
In its turn, the function $p \mapsto \mathbb E \left[\left( \frac{S_T}{F} \right)^p \right]$ is holomorphic on $D^* = \{-q^* < \mathrm{Re}(p) < p^* \}$.
From Theorem \ref{extensionFukThm}, these two functions coincide on the interval $(-p_-(\beta_-), p_+(\beta_+))$, therefore they also coincide on $D \cap D^* = D$.
\end{proof}

\section{Duality for normalized implied volatilities} \label{s:duality}

In this section, we investigate the consequences of put-call duality on the normalized implied volatilities $v^{p}(\cdot)$, $p \in [0,1]$.

Let us briefly recall the put-call duality relation and its (standard) consequences for the implied volatility $v(\cdot)$.
Let us drop the index $T$ from notation, and write $S:= S_T = F M$, where $M$ is a positive random variable with $\esp^{\Prob}[M] = 1$. 
Consider the change of measure $\frac{d\hat \Prob}{d\Prob} := M = \frac S F$. 
Set $\hat{S} = \frac F M = \frac{F^2}{S}$ (note that $\esp^{\hat \Prob}[\hat{S}] = F$).
We look for a relation between the implied volatilities of $S$ under $\Prob$ and $\hat{S}$ under $\hat{\Prob}$.
We have:
\be \label{e:putCallDuality1}
\esp^{\hat \Prob}[(\hat{S} - K)_+]
= \esp^{\Prob}\Bigl[M \Bigl(\frac F M  - K\Bigr)^+\Bigr]
= \esp^{\Prob}\Bigl[\Bigl(F - \frac S F K\Bigr)^+\Bigr]
= \frac{K}{F} \esp^{\Prob}\Bigl[\Bigl(\frac{F^2}{K} - S\Bigr)^+\Bigr]
\ee
Consider now the particular case of the Black-Scholes model where $M = \mathcal{E}_T(\sigma B)$, where $\mathcal{E}$ is the Doleans--Dade exponential, $B$ a Brownian motion and $\sigma \maj 0$.
Then, the distributions of $M$ under $\Prob$ and $\frac{1}{M}$ under $\hat \Prob$ are the same.
The equality between the outermost terms in the above chain of equalities rewrites:
\be \label{e:putCallDualityBS}
C_{\mathrm{BS}}(K, F, V) = \frac{K}{F} P_{\mathrm{BS}}\Bigl(\frac{F^2}K, F, V \Bigr)
\ee
where $C_{\mathrm{BS}}(K, F, V)$ (resp.\til$P_{\mathrm{BS}}(K, F, V) $) denotes the price of the Black-Scholes call (resp.\til put) option with strike $K$, forward value $F$,
and total implied volatility $V = \sigma \sqrt T$.
Now, by the definition of the implied volatility we have $\esp^{\hat \Prob}[(\hat{S}- K)_+] = C_{\mathrm{BS}}(K, F, \hat V(K))$ and $\esp^{\Prob}\bigl[\bigl(\frac{F^2}{K} - S \bigr)^+ \bigr] = P_{\mathrm{BS}}\bigl(\frac{F^2}K, F, V\bigl(\frac{F^2}K \bigr) \bigr)$.
Applying \eqref{e:putCallDuality1} and \eqref{e:putCallDualityBS}, we get
\[
C_{\mathrm{BS}}(K, F, \hat V(K)) = C_{\mathrm{BS}}\Bigl(K, F, V\Bigl(\frac{F^2}K \Bigr) \Bigr),
\]
therefore $\hat V(K) = V\bigl(\frac{F^2}K \bigr)$ for all $K > 0$. 
In terms of $v(k) = V(F e^k)$ (resp. $\hat v(k) = \hat V(F e^k)$), this reads
\be \label{e:dualityImpliedVol}
\hat v (k) = v(-k),
\qquad
\mbox{for all } k \in \R.
\ee

We immediately have the following:

\begin{lemma} \label{l:ualityNormalizingTransforms}
Let $f_1, f_2$ (resp.\til$\hat{f_1}, \hat{f_2}$) the first and second normalizing transformations associated to $S$ (resp.\til to $\hat{S} = \frac{F^2}{S}$).
Let $f(p,k) = p f_1(k) + (1-p) f_2(k)$, $p \in \R$, and $g(p,\cdot)$ the inverse function of $f(p,\cdot)$ for $p \in [0,1]$ (and likewise for $\hat f(p,\cdot)$ and $\hat g(p,\cdot)$).
Then 
\[
\hat f(p,k) = -f(1-p,-k) 
\qquad \mbox{for all } k \in \R, \ p \in \R.
\]
and
\[
\hat g(p,z) = -g(1-p,-z) 
\qquad \mbox{for all } z \in \R, \ p \in [0,1].
\]
In particular
\be \label{e:dualityNormalizingTransforms}
\begin{array}{ccc}
\hat{f_1}(k) = -f_2(-k), & \qquad & \hat{g_1}(z) = -g_2(-z)
\\
\hat{f_2}(k) = -f_1(-k), & \qquad & \hat{g_2}(z) = -g_1(-z).
\end{array}
\ee
\end{lemma}
\begin{proof}
The first bullet point follows immediately from \eqref{e:dualityImpliedVol} and the definition of the normalizing transformations: 
$
\hat f(p,k) = \frac{k}{\hat v(k)} + \left(\frac12-p\right) \frac{\hat v(k)}2
= -\frac{-k}{v(-k)} - \left(\frac12 - (1-p)\right) \frac{v(-k)}2 = - f(1-p,-k)$, and the claim follows.
The second bullet point is a direct consequence of the first:
$\hat	f(p,\hat g(p,z)) = z$ rewrites $-f(1-p,-\hat g(p,z)) = z$, therefore $-\hat g(p,z) = g(1-p,-z)$.
Equations \eqref{e:dualityNormalizingTransforms} are special cases	 for $p = 1$ and $p=0$ (recall that, with our notation, $f(0,\cdot) = f_2(\cdot)$).
\end{proof}

\begin{proposition} \label{p:dualityNormalizedVol}
Let $v^p$ (resp. $\hat v^p$) the $p$-normalized implied volatility associated to $S$ (resp. $\hat{S}$), as defined in \eqref{e:pNormalizedImplVolDef}.
Then,
\be \label{e:dualityNormalizedVol}
\hat v^p(z) = v^{1-p}(-z),
\qquad \mbox{for all } z \in \R.
\ee
In particular, $\hat v_1(z) = v_2(-z)$ and $\hat v_2(z) = v_1(-z)$.
\end{proposition}
\begin{proof}
Applying \eqref{e:dualityImpliedVol} and Lemma \ref{l:ualityNormalizingTransforms}, we have $\hat v^p(z) = \hat v(\hat g(p,z)) = v(-\hat g(p,z))) = v(g(1-p,-z))) = v^{1-p}(-z)$.
The last claim is a special case of \eqref{e:dualityNormalizedVol}, recalling that, with our notation, $v_1 = v^1$, $v_2 = v^0$.
\end{proof}

As an example, using formula \eqref{e:varSwapFormula} for the price of the log contract and Proposition \ref{p:dualityNormalizedVol}, we obtain the following formula for the implied variance of the $S \ln S$ contract (corresponding to the de-annualized fair strike of the Gamma Variance Swap in a diffusion model):
\[
2 \esp^{\Prob} \left[\frac{S_T}{F} \log \biggl(\frac{S_T}{F}\biggr) \right]
= 2 \esp^{\hat \Prob} \left[ -\log \biggl(\frac{\hat S_T}{F}\biggr) \right]
= \int_{\R} \hat v_2(z)^2 \phi(z) dz
= \int_{\R} v_1(-z)^2 \phi(z) dz
= \int_{\R} v_1(z)^2 \phi(z) dz,
\]
which is formula (1.2) in \cite{Fukasawa2012}.

\begin{remark}
As another application of the duality \eqref{e:dualityNormalizingTransforms}, we can deduce the invertibility of the first normalizing transformation from the invertibility of the second (or vice-versa).
Precisely, we have the following: the map $f_1$ (resp.\til $f_2$) is strictly increasing for every arbitrage-free smile $v(\cdot)$ if and only if $f_2$ (resp.\til $f_1$) is such.
\end{remark}

\begin{remark} \label{r:dualityMoments}
The formula \eqref{moment} for the moments of $S_T$ is invariant under the duality transformations $p \to 1-p$, $g_1(z) \to -g_2(-z)$ and $g_2(z) \to -g_1(-z)$.
This is consistent with Lemma \ref{l:ualityNormalizingTransforms}:
indeed, note that 
\[
M(p) =
\mathbb{E}^{\Prob} \left[\left(\frac{S_T}F \right)^p \right]
= \mathbb{E}^{\Prob} \left[ \frac{S_T}F \left(\frac{S_T}F \right)^{p-1} \right]
= \mathbb{E}^{\hat \Prob}\Biggl[ \biggl(\frac{\hat S_T}F \biggr)^{1-p} \Biggr]
=: \hat M(1-p).
\]
Equation \eqref{moment} yields, for every $p \in (-p_-(\beta_-), p_+(\beta_+))$,
\[
\hat M(1-p) = \hat L(1-p)
:= \int_{\R} \left[(1-p) e^{-p \hat g_1(z)} + p e^{(1-p) \hat g_2(z)}\right] \phi(z) dz.
\]
Applying the identities $\hat g_1(z) = -g_2(-z)$ and $\hat g_2(z) = -g_1(-z)$ in Lemma \ref{l:ualityNormalizingTransforms}, it is immediate to check that the rightmost term in the above equation coincides with $L(p)$, therefore with $M(p)$ as expected.
\end{remark}

\section{Alternative representations and Matytsin's
formula}\label{alternative-representations-and-matytsins-formula}

The formula \eqref{extensionComplexP} is written in terms of the two normalizing
transformations $g_1$ and $g_2$. We can recast it into an equivalent formula containing only the normalized implied volatility
$v_1(\cdot)$ and its derivative (or yet, a dual formula in terms of $v_2(\cdot)$). The
original formula \eqref{e:MatytsinFormula} of Matytsin \cite{Mat2000} is a special case of this representation.

Recall the definition $v_2 = v \circ g_2$ and the following relation
between $v_2$ and $g_2$: by definition, we have
$g_2(z) = k \Leftrightarrow f_2(k) = z \Leftrightarrow \frac k{v(k)} + \frac {v(k)}2 = z \Leftrightarrow k = z v(k) - \frac {v(k)^2}2$,
which yields
\begin{equation} \label{gAndV2}
g_2(z) = z v_2(z) - \frac {v_2(z)^2}2, \qquad \qquad \mbox{ for all } z \in \mathbb R. 
\end{equation}
Analogously,
\begin{equation} \label{gAndV1}
g_1(z) = z v_1(z) + \frac {v_1(z)^2}2, \qquad \qquad \mbox{ for all } z \in \mathbb R. 
\end{equation}

\begin{proposition} \label{AlternativeFormulaProp2}
For every $p \in \mathbb C$ with $\mathrm{Re}(p)\in (-p_-(\beta_-), p_+(\beta_+))$,
\begin{equation} \label{extensionComplexPg2v2}
\mathbb E \left[ e^{p X_T} \right]
= \int_{\mathbb R} e^{p \left(z v_2(z) - \frac 12 v_2^2(z) \right)} [1 - p v_2'(z)]  \phi(z) dz.
\end{equation}
\end{proposition}

\begin{proof} 
Starting from \eqref{extensionComplexP}, we have
\be \label{e:eq1}
\begin{aligned}
\mathbb E \left[\left( \frac{S_T}{F} \right)^p \right] 
&= \int_{-\infty}^{+\infty} [p e^{(p-1)g_1(z)} + (1-p) e^{p g_2(z)}] \phi(z) dz
\\
&= \int_{-\infty}^{+\infty} \left\{ p e^{(p-1)k} \phi(f_1(k)) f'_1(k) + (1-p) e^{p k} \phi(f_2(z)) f'_2(k) \right\} dk
\\
&= \int_{-\infty}^{+\infty} e^{pk} \phi(f_2(k)) [p f'_1(k) + (1-p) f'_2(k)] dk
\end{aligned}
\ee
where we used the identity $e^{-k} \phi(f_1(k)) = \phi(f_2(k))$. By definition of $f_1$ and $f_2$, we have $f_1'(k) = f_2'(k) - v'(k)$, which yields
\be \label{e:eq2}
e^{pk} \phi(f_2(k)) \left[ p f'_1(k) + (1-p) f'_2(k) \right]
= e^{pk} \phi(f_2(k)) [f'_2(k) - p v'(k)]
= e^{pk} \phi(f_2(k)) [1 - p v_2'(f_2(k))] f'_2(k)
\ee
for $\frac d{dk} v(k) = \frac d{dk} v_2(f_2(k)) = v_2'(f_2(k)) f_2'(k)$.
Plugging \eqref{e:eq2} into \eqref{e:eq1}, we obtain
\[
\mathbb E \left[\left( \frac{S_T}{F} \right)^p \right] 
= \int_{-\infty}^{+\infty} e^{pk} \phi(f_2(k)) [1 - p v_2'(f_2(k))] f'_2(k) dk
= \int_{-\infty}^{+\infty} e^{p g_2(z)} \phi(z) [1 - p v_2'(z)] dz
\]
from which the claim follows using \eqref{gAndV2}.
\end{proof}

\hspace{4mm} We have the following dual formula:

\begin{proposition}
For every $p \in \mathbb C$ with $\mathrm{Re}(p)\in (-p_-(\beta_-), p_+(\beta_+))$,
\begin{equation} \label{extensionComplexPg1v1}
\mathbb E \left[ e^{p X_T} \right] = 
\int_{\mathbb R} e^{(p-1) \left(z v_1(z) + \frac 12 v_1^2(z)\right)} [1 + (1-p) v_1'(z)]  \phi(z) dz.
\end{equation}
\end{proposition}

\begin{proof}
Follows by mimicking the proof of Prop \ref{AlternativeFormulaProp2}, now applying the identities $f_2'(k) = f_1'(k) + v'(k)$ and \eqref{gAndV1}.
\end{proof}

\begin{remark}
The two formulas \eqref{extensionComplexPg2v2} and \eqref{extensionComplexPg1v1} are the dual of each other: that is, they
are related via the 
transformations $p \leftrightarrow 1-p $, $v_1(z) \leftrightarrow v_2(-z)$.
In fact, instead of proving Equation \eqref{extensionComplexPg1v1}, we could infer it from \eqref{extensionComplexPg2v2}, applying the duality results of Section \ref{s:duality}. 
This goes as follows: use Equation \eqref{extensionComplexPg2v2} to compute $\hat M(1-p) = \esp^{\hat \Prob} \bigl[ \bigl(\frac{\hat S_T}F\bigr)^{1-p} \bigr]$ in terms of $\hat v_2(\cdot)$.
Applying the identities $\hat M(1-p) = M(p)$ from Remark \ref{r:dualityMoments} and $\hat v_2(z) = v_1(-z)$ from Proposition \ref{p:dualityNormalizedVol}, we obtain \eqref{extensionComplexPg1v1}.
\end{remark}

\begin{corollary}
Taking $p = i \eta$ for $\eta \in \mathbb R$ in
\eqref{extensionComplexPg2v2}, we obtain Matytsin's formula \cite{Mat2000}
\[
\mathbb E \left[ e^{i \eta X_T} \right] = \int e^{i \eta \left(z v_2(z) - \frac 12 v_2^2(z) \right)} [1 - i\eta v_2'(z)] \phi(z) dz.
\]
Taking $p = i \eta$ for $\eta \in \mathbb R$ in
\eqref{extensionComplexPg1v1}, we obtain the dual formula
\[
\mathbb E \left[ e^{i \eta X_T} \right] = \int e^{(i\eta-1) \left(z v_1(z) + \frac 12 v_1^2(z) \right)} [1 + (1-i\eta) v_1'(z)] \phi(z) dz.
\]
\end{corollary}

As an exercise, we can reobtain Chriss and Morokoff's formula \cite{ChrissMorok99} for
$\mathbb E[X_T] = \mathbb E \left[\log \frac{S_T} F \right]$ by computing
$\frac 1i \frac d{d \eta} \left. \mathbb E \left[ e^{i \eta X_T} \right] \right|_{\eta = 0}$.

\section{Extension of Bergomi's formula to $p \in \mathbb C$, $\mathrm{Re}(p) \in [0,1]$}

In this section, we follow Bergomi's idea \cite{Bergomi2016} of interpolating the two transformations $f_1$ and $f_2$.
When $p$ is real-valued, this procedure allows to turn Equation \eqref{moment} into (yet) another formula for $\mathbb E [(S_T/F)^p]$, now written in terms of the $p$-normalized implied volatility $v^p(\cdot)$ : the result is the compact and elegant formula \eqref{BergomiFormula} (which is also self-dual, cf.\til Remark \ref{r:BergomiDuality} below).
Starting from Proposition \ref{p:complexExtension}, we can now extend Bergomi's formula to the complex plane.

\hspace{4mm} Note that we can extend the definition of the function $f(p,\cdot)$ to $p \in \mathbb C$ in the obvious way, setting $f(p,k) = p f_1(k) + (1-p)f_2(k)$ for every $k \in \R$.
This gives 
\be \label{e:f_p_complexExtension}
\begin{aligned}
f(p,k) &= f(\Re(p),k) + i \Im(p) (f_1(k) - f_2(k)) = f(\Re(p),k) - i \Im(p) v(k)
\\
&= f(\Re(p),k) - i \Im(p) v^{\Re(p)}(f(\Re(p),k)),
\end{aligned}
\ee
so that the map $f(p,\cdot)$ is one-to-one from $\R$ onto the following curve in the complex plane:
\be \label{e:gamma_p}
\gamma_p: \R \to \mathbb C,
\qquad \gamma_p: a \mapsto a - i \Im(p) v^{\Re(p)}(a).
\ee
Consequently, for every $p$ with $\Re(p) \in [0,1]$, we can define the inverse of $f(p,\cdot)$ from $\gamma_p$ to $\R$ by
\[
g(p,\cdot): z \in \gamma_p \mapsto f(p,\cdot)^{-1}(z) = g(\Re (p),\Re(z)),
\]
where, with a slight abuse of notation, we still denote $\gamma_p$ the support $\{ \gamma_p(a): a \in \R \}$ of the curve.
Finally, we can  extend the definition of the $p$-normalized implied volatility $v^p$ to complex $p$:
\[
\mbox{for every } z \in \gamma_p,
\ v^{p}(z) := v(g(p,z)) 
\]
so that $v^{p}(z) = v^{\Re(p)}(\Re(z))$.

\begin{theorem} \label{t:BergomiExtension}
For every $p$ in the strip $\{\mathrm{Re}(p) \in [0,1]\}$,
\begin{equation} \label{e:extensionBergomiComplex}
\mathbb E \left[ e^{p X_T} \right] = 
\int_{\gamma_p}
e^{i \mathrm{Im}(p) g(p,z)}
e^{\frac 12 \mathrm{Re}(p) (\mathrm{Re}(p)-1) v^{p}(z)^2}
\phi(\Re(z)) dz,
\end{equation}
where $g(p,z) = z v^p(z) - \left(\frac12 - \mathrm{Re}(p)\right) v^p(z)^2$ and the curve $\gamma_p$ is defined in \eqref{e:gamma_p}.
\end{theorem}

\begin{proof}
In the proof of Prop. \ref{AlternativeFormulaProp2}, we have shown
that $\mathbb E \left[\left( \frac{S_T}{F} \right)^p \right] = \int_{-\infty}^{+\infty} e^{pk} \phi(f_2(k)) [p f'_1(k) + (1-p) f'_2(k)] dk$.
Therefore, using \eqref{e:f_p_complexExtension},
\be \label{e:extBergomi1}
\begin{aligned}
\mathbb E \left[\left( \frac{S_T}{F} \right)^p \right] &=
\int_{-\infty}^{+\infty} e^{pk} \phi(f_2(k)) \partial_k f(p,k) dk
\\
&= \int_{-\infty}^{+\infty} e^{pk}
\phi \Bigl( f(\mathrm{Re}(p),k) + \mathrm{Re}(p) v(k) \Bigr)
\Bigl[ 1 - i \mathrm{Im}(p)\left(v^{\mathrm{Re}(p)}\right)' (f(\Re(p),k)) \Bigr] 
\partial_k f(\Re(p),k)
dk
\\
&=
\int_{-\infty}^{+\infty} e^{p g(\Re(p),a)} \phi \left(a + \Re(p) v^{\mathrm{Re}(p)}(a) \right)
\gamma_p'(a) da
\\
&=
\int_{\gamma_p} e^{p g(p,z)} \phi \bigl(\Re(z) + \Re(p) v^{p}(z) \bigr) dz,
\end{aligned}
\ee
where we have used the identity $f_2(k) = \frac k{v(k)} + \frac12 v(k) = \frac k{v(k)} + \left(\frac12 - a\right) v(k) + a v(k) = f(a,k) + a v(k)$ in the second line.
Now, Eqs \eqref{gAndV2} and \eqref{gAndV1} are easily generalized to $g(a,z) = z v^a(z) - \left(\frac12 - a \right) v^a(z)^2$, which yields
$g(p,z) = \Re(z) v^p(z) +\Re(p) v^p(z)^2 -\frac12 v^p(z)^2$.
Plugging this last identity inside \eqref{e:extBergomi1}, after some straightforward simplifications, we obtain \eqref{e:extensionBergomiComplex}.
\end{proof}

\begin{remark}
Theorem \ref{t:BergomiExtension} shows that Matytsin's formula \eqref{e:MatytsinFormula} and Bergomi's \eqref{BergomiFormula} can be written as line integrals of the \emph{same} function on different curves in the complex plane.
Bergomi's formula is obtained setting
$\mathrm{Im}(p) = 0$ in \eqref{e:extensionBergomiComplex}: in this case, $\gamma_p = \R$.
Matytsin's formula is obtained (again) setting $\mathrm{Re}(p) = 0$ in \eqref{e:extensionBergomiComplex}: in this case, $\gamma_p = \{ a - i \Im(p) v_2(a): a \in \R\}$ (recall that $v^0(\cdot)$ corresponds to Fukasawa's $v_2(\cdot)$).
\end{remark}

Re-injecting the expression of $\gamma_p'$, formula \eqref{e:extensionBergomiComplex} can of course be written as the following (less compact) integral on the real-line
\begin{equation} \label{extensionBergomi}
\mathbb E \left[ e^{p X_T} \right] = 
\int_{\mathbb R}
e^{i \mathrm{Im}(p) g(\mathrm{Re}(p),z)}
e^{\frac 12 \mathrm{Re}(p) (\mathrm{Re}(p)-1) v^{\mathrm{Re}(p)}(z)^2 }
\Bigl[ 1 - i \mathrm{Im}(p) \frac d{dz} v^{\mathrm{Re}(p)}(z) \Bigr] \phi(z) dz,
\end{equation}
which does not require the notion of transformation $g(a,z)$ and volatility $v^a(z)$ for complex-valued $a$ and $z$.

\begin{remark} \label{r:BergomiDuality}
Just as the Equation \eqref{moment} we started from, the formulas \eqref{e:extensionBergomiComplex} and \eqref{extensionBergomi} are also self-dual: they are invariant under the transformations
$p \leftrightarrow 1-p, v^p(z) \leftrightarrow v^{1-p}(-z)$.
\end{remark}

\hspace{4mm} The restriction $\mathrm{Re}(p) \in [0,1]$ in Theorem \ref{t:BergomiExtension} is imposed by the invertibility requirement for the interpolated transformation $f(\mathrm{Re}(p),\cdot)$.
In the next section, we investigate the invertibility of this map for values of $\mathrm{Re}(p)$ that lie outside $[0,1]$ .
We will see that (apart from the trivial Black-Scholes case) there are examples of smiles for which the map $f(a,\cdot)$ remains strictly monotone on some interval (and possibly on the whole $\R$) for all $a > 1$ or all $a < 0$.

\section{The invertibility of
$k \mapsto f(p,k)$}\label{s:monotonicity}

We know (from the result of Fukasawa \cite{Fukasawa2012}) that the interpolated transformation $k \mapsto f(p,k) = p f_1(k) + (1-p) f_2(k) = \frac{k}{v(k)} + \left(\frac12-p\right) \frac{v(k)}2$ associated to \emph{any} arbitrage-free implied volatility $v(\cdot)$ is strictly increasing if $p$ is in $[0,1]$. On the other hand, in the Black-Scholes model with constant volatility parameter $\sigma$, the map $k \mapsto f(p,k) = \frac k {\sigma \sqrt T} + \Bigl(\frac12 - p\Bigr) \sigma \sqrt T$ is (linear, hence) trivially invertible, for every $p \in \mathbb R$.
It is natural, then, to wonder what happens in general -- that is, if there are other cases where $f(p,\cdot)$ is invertible when $p$ lies outside $[0,1]$, and if so, whether one can characterize (in terms of the smile, or of the underlying distribution) the values of $p$ for which this happens.

\subsection{$f(p,\cdot)$ is surjective for every $p$ in an interval
larger than $(-p_-(\beta_-), p_+(\beta_+))$}\label{fpcdot-is-surjective-for-every-p-in-an-interval-larger-than--pux5f-betaux5f--pux5fbetaux5f}

A reasonable guess would seem to conjecture that $k \mapsto f(p,k)$ is invertible when $p$ is between the critical exponents, $p \in (-p_-(\beta_-), p_+(\beta_+))$.
The following proposition shows that $f(p,\cdot)$ is actually surjective for every $p$ in an interval that is \emph{larger} than $(-p_-(\beta_-), p_+(\beta_+))$.

\begin{proposition} \label{e:f_p_surjective}
Recall that $\beta_\pm = \limsup_{k \to \pm \infty} \frac{v(k)}{\sqrt{|k|}}$. Define
\[
\tilde p_+ = \frac 1{\beta_+} + \frac 12, \qquad \tilde p_- = \frac 1{\beta_-} - \frac 12,
\]
where $\tilde p_\pm = \infty$ if $\beta_\pm = 0$.
We have $\tilde p_+ \ge p_+(\beta_+)$ resp.\til$\tilde p_- \ge p_-(\beta_-)$, where the inequalities are strict if $\beta_+ \neq 2$, resp.\til$\beta_- \neq 2$.
Moreover
\begin{itemize}
\item[-] if $p < \tilde p_+$, then $\lim_{k \to \infty} f(p,k) = +\infty$.
\item[-] if $p > -\tilde p_-$, then $\lim_{k \to -\infty} f(p,k) = -\infty$.
\end{itemize}
In particular, if $p \in (-\tilde p_-, \tilde p_+)$, the map $f(p,\cdot)$ is surjective on $\R$.
\end{proposition}

\begin{proof}
The comparison between $\tilde p_+$ and $p_+(\beta_+)$ (resp.\til $\tilde p_-$ and $p_-(\beta_-)$) is immediate to check.
Let us, then, prove the statements about the limits of $f(p,\cdot)$.

Assume $p < \tilde p_+$.
First note that, if $p \le 0$, we have $f(p,k) = \frac k{v(k)} + (\frac12 - p) v(k) \ge \frac k{v(k)} \to \infty$ as $k \to \infty$.
Then, assume $p > 0$. If $\beta_+ = 2$, there is nothing more to prove, because in this case $\tilde p_+=1$. 
If $\beta_+ \in [0,2)$, then using estimate \eqref{f2behavior}, for every $\beta_+ < \beta < 2$ we have
\be \label{f_p_estimate_right}
f(p,k) = f_2(k) - p v(k) \ge \left[\frac 1{\sqrt \beta} + \frac{\sqrt \beta}2 - p \sqrt{\beta} \right] \sqrt k,
\qquad \forall k > \overline k_{\beta},
\ee
for some $\overline k_{\beta} \maj 0$.
The condition $p < \tilde p_+$ entails $\frac 1{\sqrt \beta_+} + \frac{\sqrt \beta_+}2 - p \sqrt{\beta_+} > 0$.
Therefore, when $\beta$ is sufficiently close to $\beta_+$, the coefficient in front of $\sqrt k$ in \eqref{f_p_estimate_right} is positive, so that $\lim_{k \to \infty} f(p,k) = +\infty$.

The symmetric argument holds when $p > -\tilde p_-$.
Let us provide the details: First, if $p \le 1$, then $f(p,k) \le \frac k{v(k)} \to -\infty$ as $k \to -\infty$.
Then, assume $p < 1$. Once again, if $\beta_- = 2$, there is nothing more to prove, because in this case $\tilde p_- = 0$.
If $\beta_- \in [0,2)$, then using estimate \eqref{f2secondEstimate}, for every $\beta_- < \delta < 2$ we have
\[
f(p,k) = f_1(k) + (1-p) v(k) \le \left[-\biggl( \frac 1{\sqrt \delta} + \frac{\sqrt \delta}2 \biggr) + (1-p) \sqrt{\delta} \right] \sqrt{|k|},
\qquad \forall k < \underline k_{\delta},
\]
for some $\underline k_{\delta} \mino 0$.
The condition $p > -\tilde p_-$ entails $-\frac 1{\sqrt \delta} - \frac{\sqrt \delta}2 + (1-p) \sqrt{\delta} < 0$.
Therefore, when $\delta$ is sufficiently close to $\beta_-$, the coefficient in front of $\sqrt{|k|}$ is negative, so that $\lim_{k \to -\infty} f(p,k) = -\infty$.
\end{proof}

\subsection{Some results on the monotonicity of $f(p,\cdot)$}

In this section we show that there exist (non trivial, and practically interesting) situations where the map $f(p,\cdot)$ is strictly monotone for values of $p$ that lie outside $[0,1]$.
We start with the following lemma.

\begin{lemma} \label{lemmaPartialMonotonicity}
On the set $\{k : v'(k) \le 0 \}$, we have $\partial_k f(p,k) > 0$, for every $p \maj 1$.

On the set $\{k : v'(k) \ge 0 \}$, we have $\partial_k f(p,k) > 0$, for every $p \mino 0$.
\end{lemma}

The first case in Lemma \ref{lemmaPartialMonotonicity} is relevant in particular in Equity markets, where smiles on stock indices are often monotonically decreasing on the observed interval of strikes (in particular for larger maturities).

\begin{proof}
We first focus on the case $p \maj 1$.
Using the identities
\begin{equation} \label{f1f2derivatives}
\begin{aligned}
f_1'(k) &= \frac 1{v(k)} - \frac k{v(k)^2} v'(k) - \frac{v'(k)}2 = \frac 1{v(k)} \left(1 - v'(k) f_2(k) \right)
\\
f_2'(k) &= \frac 1{v(k)} - \frac k{v(k)^2} v'(k) + \frac{v'(k)}2 = \frac 1{v(k)} \left(1 - v'(k) f_2(k) \right) + v'(k)
\end{aligned}
\end{equation}
we have
\[
\begin{aligned}
\partial_k f(p,k) = p f_1'(k) +(1-p) f_2'(k)
&= \frac 1{v(k)} \left[ p(1 - v'(k) f_2(k)) + (1-p)(1 - v'(k) f_2(k)) \right] + (1-p) v'(k)
\\
&= \frac 1{v(k)} \left[ 1 - v'(k) f_2(k) \right] + (1-p) v'(k).
\end{aligned}
\]
We know from \cite[Lemma 2.6]{Fukasawa2012} that $1 - v'(k) f_2(k) > 0$ for all $k \in \mathbb R$. The conclusion follows.

The case $p \mino 0$ follows from duality (Section \ref{s:duality}): by Lemma \ref{l:ualityNormalizingTransforms}, we have $f(p,k) = -\hat f(q,-k)$, therefore $\partial_k f(p,k) = \partial_k \hat f(q,-k)$, where $q = 1-p$ is larger than $1$.
In the first part of the proof, we have shown that $\partial_k \hat f(q,l) 	\maj 0$ on $\{l: \hat v'(l) \le 0 \}$. 
Since $\hat v(k) = v(-k)$ from Eq \eqref{e:dualityImpliedVol}, we have $\{l: \hat v'(l) \le 0 \} = - \{k: v'(k) \ge 0 \}$, and the second claim follows.
\end{proof}


\begin{proposition} \label{p:monotoneSmile_f_invertible}
$i)$ If the implied volatility $k \mapsto v(k)$ is decreasing, then $\tilde p_+ = p_+({\beta_+}) = \infty$ and the function $f(p,\cdot)$ is invertible from $\R$ to $\R$, for all $p \ge 0$.

$ii)$ If the implied volatility $k \mapsto v(k)$ is increasing, then $\tilde p_- = p_-({\beta_-}) =\infty$ and the function $f(p,\cdot)$ is invertible from $\R$ to $\R$, for all $p \le 1$.
\end{proposition}

Proposition \ref{p:monotoneSmile_f_invertible} allows us to extend the definition of the normalized implied volatility $v^p(\cdot)$ in Eq \eqref{e:pNormalizedImplVolDef} to every $p \ge 0$ if we start from a decreasing smile $v(\cdot)$ (resp.\til to every $p \le 1$ if we start from an increasing smile).
As an application:

\begin{corollary}
The extended Bergomi's formula \eqref{t:BergomiExtension} holds for all
$\{p: \mathrm{Re}(p) \ge 0\}$ in case i) of Proposition \ref{p:monotoneSmile_f_invertible}, resp.\til$\{p: \mathrm{Re}(p) \le 1\}$ in case ii).
\end{corollary}

\begin{proof}[Proof of Proposition \ref{p:monotoneSmile_f_invertible}]
Let us consider the first case ($v(\cdot)$ is increasing).
We have to prove the statement for $p\maj 1$, for we already know that $f(p,\cdot)$ is invertible for $p \in [0,1]$.
It follows from Lemma \ref{lemmaPartialMonotonicity} that $f(p,\cdot)$ is strictly monotone for every $p \maj 1$. 
Since, by assumption, $v(k)$ has a finite limit for $k \to \infty$, we have $\lim_{k \to \infty}\frac{v(k)}{\sqrt k} = \beta_+ = 0$, therefore $\tilde p_+ = p_+({\beta_+}) = \infty$.
It then follows from Proposition \ref{e:f_p_surjective} that $f(p,\cdot)$ is surjective on $\R$ for every $p \maj 1$, and the claim is proved.
The second case ($v(\cdot)$ is decreasing) is proven in the same way, using Lemma \ref{lemmaPartialMonotonicity} and Proposition \ref{e:f_p_surjective} with $\tilde p_- = \infty$.
\end{proof}

\hspace{4mm} \textbf{A class of smiles for which $f(p,\cdot)$ is not invertible when $p \notin [-\tilde p_-, \tilde p_+]$}.
We now exhibit a large class of implied volatilities (with finite critical moments, hence finite coefficients $\tilde p_\pm$) for which the function $f(p,\cdot)$ fails to be both monotone and surjective when $p$ is outside the interval $(-\tilde p_-, \tilde p_+)$.
Consider the following assumption
\[
\mathrm{(H)} \qquad v'(k) \sim \pm \frac{\sqrt{\beta_\pm}}{2 \sqrt{|k|}},  \qquad k \to \pm \infty,
\qquad \mbox{where } \beta_\pm \neq 0.
\]
By de L'Hopital's rule, $\mathrm{(H)}$ implies
$v(k) \sim \sqrt{\beta_\pm |k|}$ as $k \to \pm \infty$, therefore we have \eqref{e:wingSlopes}.
Note that the SSVI parameterisation \eqref{SSVI} with $\varphi \maj 0$ satisfies assumption $\mathrm{(H)}$ (recall Eq \eqref{e:SSVI_slopes}). 

Using the definition of $f(p,k)$, it is straightforward to check that assumption (H) implies
\[
f(p,k) \sim \sqrt{\beta_+} \left(\tilde p_+ - p \right) \sqrt{k}
\quad
\mbox{ as } k \to \infty,
\qquad
f(p,k) \sim -\sqrt{\beta_-} \left(\tilde p_- + p \right) \sqrt{|k|}
\quad
\mbox{ as } k \to -\infty.
\] 
The limits of $f(p,k)$ for $k \to \pm \infty$ are then easily assessed: the case $p \in (-\tilde p_-, \tilde p_+)$ is already considered in Proposition \ref{e:f_p_surjective}; if $p$ lies outside the closure of this interval, then both limits have the same sign:
\be \label{e:limits_f_p_non_surjective}
\begin{array}{lll}
p \maj \tilde p_+ & \Rightarrow & \lim_{k \to \pm \infty} f(p,k) = -\infty,
\\
p \mino -\tilde p_- & \Rightarrow & \lim_{k \to \pm \infty} f(p,k) = \infty.
\end{array}
\ee
In particular we see that, being continuous, the map $f(p,\cdot)$ cannot be surjective on $\R$ when $p \notin [-\tilde p_-, \tilde p_+]$.

\begin{proposition} \label{p:f_p_not_monotone}
Assume $\mathrm{(H)}$. Then, if $p > \tilde p_+$ or $p \mino -\tilde p_-$, the map $f(p,\cdot): \R \to \R$ is neither monotone, nor surjective.

More precisely: If $p > \tilde p_+$ (resp.\til$p \mino -\tilde p_-$), there exist $\underline k$ and $\overline k$ such that $f(p,\cdot)$ is strictly increasing (resp.\til decreasing) on $(-\infty,\underline k)$ and strictly decreasing (resp.\til increasing) on $(\overline k, \infty)$.
\end{proposition}

A numerical example of the situation described in Proposition \ref{p:f_p_not_monotone} will be given in the next section -- see Figure \ref{f:f_p_SSVI}.

\begin{proof}[Proof of Proposition \ref{p:f_p_not_monotone}]
The statement about the surjectivity of $f(p,\cdot)$ has already been proven above (recall \eqref{e:limits_f_p_non_surjective}); we prove that $f(p,\cdot)$ is not monotone.

Let us first consider the case $p \maj \tilde p_+$.
The condition $v'(k) \sim -\frac{\sqrt {\beta_{-}}}{2 \sqrt{|k|}}$ as $k \to -\infty$ implies that $v'$ is negative on the half-line $(-\infty,\underline k)$, for some $\underline k$.
It follows from Lemma \ref{lemmaPartialMonotonicity} that $f(p,\cdot)$ is strictly increasing on $(-\infty,\underline k)$.

Using the first equation in \eqref{f1f2derivatives}, together with the identity $f_2'(k) = \frac 1{v(k)} \left(1 - v'(k) f_1(k) \right) + v'(k)$, we have
\be \label{e:f_p_derivative}
\begin{aligned}
\partial_k f(p,k) = \frac 1{v(k)} \left[ p(1 - v'(k) f_2(k)) + (1-p)(1 - v'(k) f_1(k)) \right]
&= \frac 1{v(k)} \left[ 1 - v'(k) (p f_2(k)) + (1-p) f_1(k)) \right]
\\
&= \frac 1{v(k)} \left[ 1 - v'(k) f(1-p,k) \right].
\end{aligned}
\ee
It follows from assumption $\mathrm{(H)}$ that
\be \label{e:f_oneMinusP_asymptotics}
f(1-p,k) = \frac k{v(k)} + \biggl(p-\frac12\biggr) v(k) \sim \biggl(\frac 1{\sqrt{\beta_+}} + \biggl(p-\frac12 \biggr) \sqrt{\beta_+} \biggr) \sqrt k
\qquad
\mbox{ as }k \to \infty,
\ee
therefore
\[
\lim_{k \to \infty} v'(k) f(1-p,k) =  \frac{\sqrt{\beta_+}}2 \left(\frac 1{\sqrt{\beta_+}} + \biggl(p-\frac12\biggr) \sqrt{\beta_+} \right)
= \frac 12 + \biggl(p-\frac12\biggr) \frac{\beta_+}2 =: A(p,\beta_+).
\]
It is immediate to see that $\mathrm{sign}(A(p,\beta_+) - 1) = \mathrm{sign}(p - \tilde p_+)$.
Consequently, if $p > \tilde p_+$, it follows from \eqref{e:f_p_derivative} that $\partial_k f(p,k)$ is negative for $k$ large enough, and the claim on the intervals of monotonicity of $f(p,\cdot) $is proven.

We can now deduce the claim in the case $p \mino -\tilde p_-$ from duality: consider the dual implied volatility $\hat v(\cdot)$ defined in Section \ref{s:duality}.
It follows from \eqref{e:dualityImpliedVol} and assumption (H) that $\hat v'(k) \sim \frac{\sqrt{\beta_-}}{2 \sqrt{k}}$ as $k \to \infty$ and $\hat v'(k) \sim -\frac{\sqrt{\beta_+}}{2 \sqrt{|k|}}$ as $k \to -\infty$, so that $\hat \beta_\pm = \beta_\mp$ (the duality transformation exchanges the right and left slopes of the smile).
Denote $\hat p_+$ the coefficient $\tilde p_+$ associated to $\hat v$, that is: $\hat p_+ = \frac 1{\hat \beta_+} + \frac 12 = \frac1{\beta_-} + \frac 12 = \tilde p_- + 1$. 
We know from Lemma \ref{l:ualityNormalizingTransforms} that $f(p,k) = -\hat f(1-p,-k)$. Since $p \mino -\tilde p_-$, then $1-p \maj 1 + \tilde p_- = \hat p_+$.
In the first part of the proof, we have proven that $\hat f(1-p,k)$ is strictly increasing on the half-line $(-\infty, \underline k)$ for some $\underline k$, and strictly decreasing on $(\overline k, \infty)$ for some $\overline k$, therefore the claim on the monotonicity of $f(p,\cdot)$ follows.
\end{proof}

\hspace{4mm} In view of Propositions \ref{e:f_p_surjective}, \ref{p:monotoneSmile_f_invertible} and \ref{p:f_p_not_monotone}, it seems reasonable to conjecture that the map $f(p,\cdot): \R \to \R$ is invertible if $p \in (-\tilde p_-,\tilde p_+)$, and only if $p \in [-\tilde p_-,\tilde p_+]$.
Leaving the proof of this statement for future work, we numerically check this fact on an arbitrage-free SSVI parameterisation in the next section.

\section{The case of the SSVI parameterisation}\label{the-ssvi-parameterisation}

Recall the SSVI parameterisation \eqref{SSVI}, where, for fixed $T>0$,
\[
\theta_T = \theta >0,
\qquad \varphi(\theta_T) = \varphi \ge 0,
\qquad \rho \in (-1,1).
\]

\subsection{No-arbitrage conditions}\label{no-arbitrage-conditions}

Theorem 4.2 in \cite{SSVI} proves that the implied variance $v_{\mathrm{SSVI}}^2(k)$ is free of arbitrage (for the given maturity $T$) if the following conditions are satisfied:

\begin{condition} \label{noArbSSVIslice}
\[
\begin{array}{l l}
(1) & \theta \varphi (1+|\rho|) < 4
\\
(2) & \theta \varphi^2 (1+|\rho|) \le 4.
\end{array}
\]
\end{condition}

Moreover, \cite[Lemma 4.2]{SSVI} shows that the condition
$\theta \varphi (1+|\rho|) \le 4$ is necessary.
The inequality 2) in Condition \ref{noArbSSVIslice} is sufficient, but not necessary.
\\
We cross-check below that, under Condition \ref{noArbSSVIslice}, our standing Assumptions \ref{standingAssumptionsV} (i) and (ii) on $v$ are (as one could expect) satisfied.
We also further discuss the limiting case $\theta \varphi (1+|\rho|)=4$.

\subsection{The limiting case $\theta \varphi (1+|\rho|)=4$ and
checking Assumption \ref{standingAssumptionsV} (ii) on
$v$}\label{the-limiting-case-theta-varphi-1rho4-and-checking-condition-ii-on-v}

We show that the case $\rho \ge 0$ and $\theta \varphi (1+|\rho|)=4$ is not arbitrage-free.
Therefore, if $\rho \ge 0$, Condition \ref{noArbSSVIslice}(1) (with strict inequality) is also necessary.
When $\rho < 0$, we show that the case $\theta \varphi (1+|\rho|)=4$ is ruled out by our Assumption \ref{standingAssumptionsV} (ii) of zero mass at $K=0$.

	Assume $\theta \varphi (1+|\rho|)=4$. We separate the two cases
$\rho \ge 0$ and $\rho < 0$: assume first that $\rho \ge 0$. It is easy
to see that
$\lim_{k \to \infty} \frac{v_{\mathrm{SSVI}}^2(k)}k = \theta \varphi \frac{1+\rho}2 = 2$.
Then we can compute, for $k > 0$:
\[
\begin{aligned}
v_{\mathrm{SSVI}}^2(k) - 2k = k \left[\frac{v^2(k)}k - 2 \right]
&= k \left[ \frac{\theta}2 \left(\frac 1k + \rho \varphi + \sqrt{\varphi^2 + 2 \frac{\varphi \rho}k + \frac 1{k^2} } \right) - 2 \right]
\\
&= k \left[ \frac{\theta}2 \left(\frac 1k + \rho \varphi + \varphi + \frac \rho k + O\Bigl(\frac 1{k^2}\Bigr) \right) - 2 \right]
\\
&= k \left[ \frac{\theta}{2k}  + O\Bigl(\frac 1{k^2}\Bigr) \right]
=
\frac{\theta}2 (1 + \rho) + O\Bigl(\frac 1{k}\Bigr)
\to \frac{\theta}2 (1 + \rho) > 0,
\qquad \mbox{as } k \to \infty.
\end{aligned}
\]
The limit above contradicts the following property from \cite[Lemma
3.1]{Lee04}: if $v$ is an arbitrage-free smile, there exists $\overline k$
such that $v^2(k) - 2k < 0$ for every $k > \overline k$.\footnote{This property was subsequently improved to
$\lim_{k \to \infty}(v^2(k) - 2k) = -\infty$ by Rogers and Tehranchi \cite{RogTeh}.}

Now assume $\rho < 0$.
The analogous computation for negative $k$ gives
\[
v_{\mathrm{SSVI}}^2(k) - 2|k| = \frac{\theta}2 (1 - \rho) + O\Bigl(\frac 1{|k|}\Bigr) \to \frac{\theta}2 - \rho,
\qquad \mbox{as } k \to -\infty.
\]
In general, a positive value of $\lim_{k \to -\infty} (v^2(k) - 2|k|)$ is
not in contradiction with no-arbitrage: indeed, we know from \cite[Propositions 2.4 and 2.5]{Fukasawa2010} that an arbitrage-free implied volatility satisfies
\[
\lim_{k \to -\infty} \bigl( v(k) - \sqrt{2|k|} \bigr) = N^{-1}(\mathbb P(S_T = 0)),
\]
where the right hand side is worth $-\infty$ if $\mathbb P(S_T = 0)$ (see also \cite[Thm 3.6]{ImpliedVolAtom}).
As discussed in the Introduction, our Assumption \ref{standingAssumptionsV} (ii) on the coefficient $d_2(k,v(k))$ is equivalent to $\mathbb P(S_T = 0) = 0$, therefore $v(k) - \sqrt{2|k|} \to -\infty$ as $k \to -\infty$.
Consequently, $v^2(k) - 2|k| = \bigl( v(k) - \sqrt{2|k|} \bigr) \bigl( v(k) + \sqrt{2|k|} \bigr)$ tends to $-\infty$ as well, and the case $\theta \varphi (1+|\rho|)=4$ and $\rho \le 0$ is ruled out by Assumption \ref{standingAssumptionsV} (ii).

Finally, a slight modification of the computation
above allows to see that, if $\theta \varphi (1+|\rho|)< 4$, we have $v_{\mathrm{SSVI}}^2(k) - 2|k| \sim -a|k|$ as $k \to -\infty$ with $a = \frac12 \bigl(4 - \theta \varphi (1+|\rho|) \bigr) \maj 0$,
therefore Assumption \ref{standingAssumptionsV} (ii) is satisfied when Condition \ref{noArbSSVIslice}(1) is in force.

\subsection{Checking Assumption \ref{standingAssumptionsV} (i) on $v$}\label{checking-condition-i}

Denote $w(k) = v_{\mathrm{SSVI}}(k)^2$. We can compute
\be \label{e:derivativesSSVI}
w'(k) = \frac{\theta\varphi}{2} \left( \rho + \frac{\varphi k + \rho}{\sqrt{(\varphi k + \rho)^2 + 1 -\rho^2}} \right), 
\qquad
w''(k) = \frac{\theta\varphi^2}{2}\frac{1-\rho^2}{((\varphi k + \rho)^2 + 1 -\rho^2)^{3/2}}.
\ee
Note that the second equation shows that $w(\cdot)$ is a convex function.
If $\varphi = 0$, $w(\cdot)$ is identically equal to $\theta \maj 0$ (so that this case corresponds to Black-Scholes implied volatility). Let us assume, then, $\varphi > 0$.
From the first equation in \eqref{e:derivativesSSVI}, $w'(k) = 0$ if and only if
$k = -\frac{2 \rho} \varphi =: k_{\min}$. At this point, we have
\[
w(k) \ge w(k_{\min}) = \frac \theta 2 \left( 1 - 2 \rho^2 + \sqrt{3 \rho^2 +1} \right),
\qquad \forall k \in \mathbb R.
\]
A straightforward computation shows that the function of $\rho$ in the RHS above is strictly positive, for any $\rho \in (-1,1)$. This guarantees that
$v(k) = \sqrt{w(k)} > 0, \ \forall k \in \mathbb R$.

Moreover, the argument of the square-root function in \eqref{SSVI} being lower bounded by $1 -\rho^2 > 0$, we have $w \in C^2(\mathbb R)$
(actually, $w \in C^{\infty}(\mathbb R)$).
Since $w(k) \ge w(k_{\min}) > 0$, we also have that
$k \mapsto v(k) = \sqrt{w(k)}$ is $C^2$.
Overall, Assumption \ref{standingAssumptionsV} (i) is also satisfied.

%
%

\subsection{Numerical tests}

It is immediate to check that the set of parameters
\be \label{SSVIparams}
\theta = (0.25)^2 = 0.0625, \quad \rho = -0.8, 
\quad \varphi = 1.40
\ee
satisfies Condition \eqref{noArbSSVIslice}.
Figure \ref{f:SSVI_f1_f2} shows the resulting SSVI implied volatility smile and the two corresponding transformations $f_1$ and $f_2$, on a large interval $(k_{min}, k_{max})$.

\begin{figure}[t]
\includegraphics[height=0.278\textheight, width=0.5\textwidth]{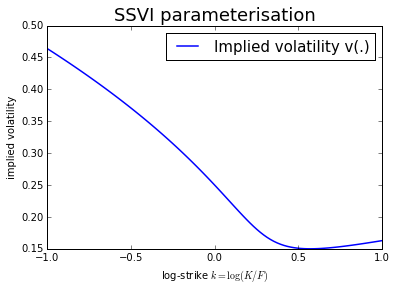}
\includegraphics[width=0.5\textwidth]{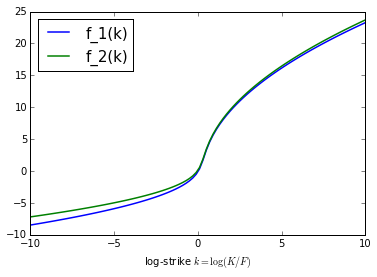}
\caption{Left: SSVI parameterisation \eqref{SSVI} of the implied volatility $v(\cdot)$ with the arbitrage-free parameters in \eqref{SSVIparams}. Right: the induced transformations $f_1, f_2$.}
\label{f:SSVI_f1_f2}
\end{figure}

Recall that $\beta_\pm(\mathrm{SSVI}) = \theta \varphi \frac{1\pm \rho}2$.
In Figure \ref{f:f_p_SSVI}, we compute the values of the coefficients $\tilde p_{\pm}$ in Proposition \ref{e:f_p_surjective}, and plot the function $k \mapsto f(p,k)$ for different values
of $p$.
As predicted by Proposition \ref{p:f_p_not_monotone}, one can see (and we did check on the numerical values) that $f(p,\cdot)$ is not increasing anymore for large (resp.\til small) $k$ when $p$ is larger than $\tilde p_+$ (resp.\til $p$ is smaller than $-\tilde p_-$).
On the contrary, $f(p,\cdot)$ does appear to be strictly increasing (at least on the considered interval of log-strikes) for $p$ within the interval $(-\tilde p_-, \tilde p_+)$.

\begin{figure}[t]
\includegraphics[height=0.3\textheight, width=0.5\textwidth]{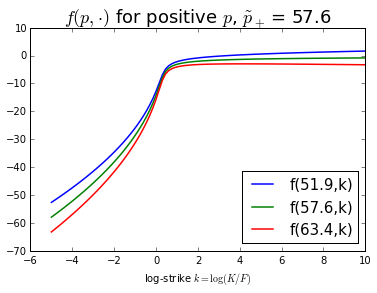}
\includegraphics[height=0.3\textheight,width=0.5\textwidth]{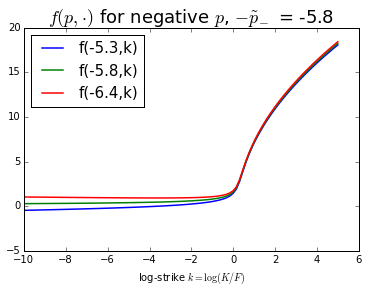}
\caption{Plot of the function $k \mapsto f(p,k)$ induced by the SSVI parameterisation \eqref{SSVI} with the parameters in \eqref{SSVIparams}, for different values of $p$ around the thresholds $\tilde p_+$ (left pane) and $-\tilde p_-$ (right pane) from Proposition \ref{e:f_p_surjective}.}
\label{f:f_p_SSVI}
\end{figure}

\section{Recovering the Black-Scholes formula}\label{recovering-black-scholes}

Having derived formulas for the extended characteristic function of the log-price, prices of European options on $S_T$ can (also) be recovered with standard transform-based methods.
As a consistency check, we show that the formula $P(K) = K N(f_2(k)) - F N(f_1(k))$ for a put option (the Black-Scholes formula) can be restored from Theorem \ref{extensionFukThm}.
We assume $B(0,T)=1$ and $F=1$ for simplicity.
We apply the following inversion theorem (see e.g.~\cite{LeeTransforms04}):


\begin{theorem} \label{t:FourierInversion}
Denote $\varphi_T(u) = \esp[e^{iuX_T}]$ the characteristic function of the log-price.
Then, for every $K>0$,
the price of a put option with strike $K$ and maturity $T$ is given by:
\[
P(K) =
R_{\alpha} + 
\frac{1}{2 \pi} \int_{-\infty}^{+\infty}
\frac{K^{-\alpha-iu+1}}{(\alpha+iu) (\alpha-1+iu)}
\varphi_T(u-i\alpha) du,
\]
where
\[
R_{\alpha} = \left\{
\begin{array}{ll}
K - 1 & \mbox{if } 1 \mino \alpha \mino p^*
\\
K & \mbox{if } 0 \mino \alpha \mino 1
\\
0 & \mbox{if } -q^* \mino \alpha \mino 0.
\end{array}
\right.
\]
\end{theorem}

Applying Theorem \ref{t:FourierInversion} and Proposition \ref{p:complexExtension} we get, choosing $\alpha \in (0,1)$ and using $i(u-i\alpha)=iu+\alpha$:
\begin{multline*}
P(K) = K 
\\
+ \frac{1}{2 \pi} \int_{-\infty}^{+\infty} \frac{K^{-\alpha-iu+1}}{(\alpha+iu) (\alpha-1+iu)}
\int_{-\infty}^{+\infty}
\left [(\alpha+iu) e^{g_1(z)(\alpha+iu-1)} + (1-\alpha-iu) e^{(\alpha+iu) g_2(z)} \right] \phi(z) dz \, du.
\end{multline*}
Using Fubini's Theorem,
\begin{multline*}
P(K) = K
\\ + \frac{1}{2 \pi} \int_{-\infty}^{+\infty} \phi(z) dz
\int_{-\infty}^{+\infty} \frac{K^{-\alpha-iu+1}}{(\alpha+iu) (\alpha-1+iu)}
\left[ (\alpha+iu) e^{g_1(z)(\alpha+iu-1)} + (1-\alpha-iu) e^{(\alpha+iu) g_2(z)} \right] du
\end{multline*}
or yet, after simplification:
\be \label{e:putMellin} 
P(K) = K +\frac{1}{2 \pi} \int_{-\infty}^{+\infty} \phi(z) dz
\int_{-\infty}^{+\infty}
\left [\frac{K^{-\alpha-iu+1} e^{g_1(z)(\alpha+iu-1)} }{\alpha-1+iu}
- \frac{K^{-\alpha-iu+1} e^{(\alpha+iu) g_2(z)} } {\alpha+iu} \right] du.
\ee
Set $k = \log{K},\; a(z) = k-g_1(z),\; b = k-g_2(z)$. Then
\be \label{e:putRepr}
P(K) = K + \frac{1}{2 i \pi} \int_{-\infty}^{+\infty} \phi(z) \bigl[ I(a(z), \alpha-1) - K I(b(z), \alpha) \bigr] dz
\ee
where
\[
I(c,d ) = \int_{d-i\infty}^{d+i\infty} \frac{e^{-c \omega}}{\omega} d\omega.
\]
We show now: 
\begin{lemma} \label{l:residue}
%
%
%

If $c<0$ and $d <0$, then $ I(c, d) = 0$.

If $c<0$ and $d >0$, then $ I(c, d) = 2 i\pi$.

If $c>0$ and $d <0$, then $ I(c, d) = -2 i\pi$.

If $c>0$ and $d >0$, then $ I(c, d) = 0$.
\end{lemma}

\begin{proof}
Consider first the case $c \maj 0$, and the rectangle defined by the real coordinates  $d$ and $d+R_1$ and the imaginary coordinates $-R_2, R_2$ wih  positive $R_1$ and $R_2$.
By Cauchy residue formula, the integral of the function $\omega \to \frac{e^{-c \omega}}{\omega}$ over this (clockwise) rectangle contour is equal to $0$ if $d>0$, or to the residue at the pole $0$, which is $1$, times $-2i\pi$ because we go clockwise, if the pole $0$ is inside the rectangle, i.e.\til if $d<0$.
Now take $R_1$ to $\infty$: since $c$ is positive, the integral on the right segment goes to zero, and the integral on the bottom and top segments are absolutely convergent. By Lebesgue dominated convergence theorem, those integrals go to zero as $R_2$ goes to infinity.
Since the remaining integral is exactly $I(c,d)$, we have proven the last two statements.	
The proof of the case $c<0$ is exactly the same, working with the real coordinates $d$ and $d-R_1$ instead, and running anticlockwise on the rectangle contour.
\end{proof}

Since $\alpha \in (0,1)$, the contribution of the first term on the RHS of \eqref{e:putRepr} is given by the case $d=\alpha-1<0$, so that
\[
\frac{1}{2 i \pi} \int_{-\infty}^{+\infty} \phi(z) I(a(z), \alpha-1)
=
-\int_{-\infty}^{+\infty} \phi(z) 1_{\{a(z)>0\} } dz
= -\int_{g_1(z)<k} \phi(z) dz = -N(f_1(k)).
\]
For the second term, we are in the case $d=\alpha>0$, so that the contribution is 
\[
-K \frac{1}{2 i \pi} \int_{-\infty}^{+\infty} \phi(z) I(b(z), \alpha) dz
=
- K \int_{-\infty}^{+\infty} \phi(z) 1_{ \{b(z) \mino 0\} } dz
=
- K \int_{g_2(z) \maj k} \phi(z) dz = - K + K N(f_2(k)).
\]
Summing up, we have recovered the Black-Scholes formula for the put option:
\[
P(K)
=
K N(f_2(k)) - N(f_1(k))
=
\mathrm{Put_{BS}}(k, v(k)).
\]

\appendix
\section{Appendix} \label{s:appendix}

\subparagraph{Proof of Lemma
\ref{auxiliaryIntegrability}}

We focus on the case $p > 0$. As explained above, the statement of the Lemma is true for every $p \in (0,1) = (-p_-(2), p_+(2))$.
Therefore, we can limit ourselves
to $\beta_+ \in [0,2)$,
and assume $p \maj 1$.

By definition of $\limsup$, we have the following: for every
$\beta > \beta_+$, there exists $\overline k_\beta > 0$ such that
$v(k) < \sqrt{\beta k}$ for all $k > \overline k_\beta$. We claim that
this implies
\begin{equation} \label{f2behavior}
\mbox{For every }  \beta_+ < \beta < 2, \ f_2(k) \ge \left(\frac 1{\sqrt{\beta}} + \frac{\sqrt{\beta}}2 \right) \sqrt{k}, \ \ \forall k > \overline k_\beta.
\end{equation}
It follows from \eqref{f2behavior} that
\[
z = f_2(g_2(z)) \ge \left(\frac 1{\sqrt{\beta}} + \frac{\sqrt{\beta}}2 \right) \sqrt{g_2(z)}, \qquad \mbox{for $z$ large enough,}
\]
which entails
$g_2(z) \le \frac{z^2}{\left(\frac 1{\sqrt{\beta}} + \frac{\sqrt{\beta}}2 \right)^2} = \frac{z^2}{2 p_+(\beta)}$.
On the other hand, using $f_1(z) = f_2(z) - v(z)$, we obtain
$f_1(k) \ge \left(\frac 1{\sqrt{\beta}} - \frac{\sqrt{\beta}}2 \right) \sqrt{k}$
for all $k > \overline k_\beta$, therefore
$g_1(z) \le \frac{z^2}{2 p_-(\beta)}$ for $z$ large enough. For such
$z$, we have \[
e^{(p-1) g_1(z)} \phi(z) \le \exp \left(z^2 \frac{p-1}{2 p_-(\beta)} - \frac {z^2}2 \right)
= \exp \left(z^2 \frac{p - p_+(\beta)}{2 p_-(\beta)} \right)
\] and \[
e^{p g_2(z)} \phi(z) \le \exp \left(z^2 \frac{p}{2 p_+(\beta)} - \frac {z^2}2 \right)
= \exp \left(z^2 \frac{p - p_+(\beta)}{2 p_+(\beta)} \right).
\]
Recall that $p_\pm(\beta) > 0$ for every $\beta \mino 2$.
Since $p - p_+(\beta_+) < 0$ by assumption, taking $\beta$ sufficiently close to
$\beta_+$, we have $p - p_+(\beta) < 0$, too.
Therefore, using the last two estimates above, we obtain that the functions in \eqref{integrability} are integrable at $+\infty$.
Using the fact that $g_1(z), g_2(z) \to -\infty$ as $z \to -\infty$ and $p \ge 1$, we
obtain that the functions in \eqref{integrability} are also integrable at $-\infty$.

The case $p < 0$ is proven analogously.

\emph{Proof of \eqref{f2behavior}} : We proceed along the lines of
\cite[Lemma 2.6]{Fukasawa2010} (note that we are referring here to an ArXiv preprint: this lemma was not reported in the published version of the article). For every $a > 1$ and $k > \overline k_\beta$, we have

\begin{equation} \label{f2firstEstimate}
f_2(k) = \frac{k}{v(k)} + \frac{v(k)}2 = \frac{k}{v(k)} + \frac{a v(k)}2 - \frac{(a-1)v(k)}2 \ge \sqrt{2ak} - \frac{(a-1)\sqrt{\beta k}}2
\end{equation}

where the last inequality holds because the arithmetic mean is larger
than the geometric mean. Estimate \eqref{f2behavior} then follows by choosing $a = 2/\beta$ (which in fact provides the optimal lower bound in \eqref{f2firstEstimate}).
\qed

\subsection{Proof of Theorem \ref{extensionFukThm}}

We first prove a weaker version of Theorem \ref{extensionFukThm}:

\begin{proposition} \label{p:extensionFukThmWeaker}
Assume 
\be \label{pBoundsStronger}
\max(-p_-(\beta_-),-q^*) < p < \min(p_+(\beta_+), p^*).
\ee
Then, Equation \eqref{extensionFukasawa} holds for every absolutely continuous function $\Psi$ such that $\Psi$ and $\Psi'$ have exponential growth of order $p$.
\end{proposition}

In order to prove Proposition \ref{p:extensionFukThmWeaker}, we need the following intermediate result.

\begin{lemma}  \label{auxiliaryIntegrabilityAndLimits}
Assume that $p$ satisfies \eqref{pBoundsStronger}, and let $\Psi$ be a function with exponential growth of order $p$. 
Then, the function $k \mapsto \Psi(k) v'(k)\phi(f_2(k))$ is integrable on $\mathbb R$, and satisfies $\lim_{k \to \pm \infty} \Psi(k) v'(k)\phi(f_2(k)) = 0$.
\end{lemma}

In the proof of Lemma \ref{auxiliaryIntegrabilityAndLimits}, we will make use of the identity
\begin{equation} \label{vPrime}
\phi(f_2(k)) v'(k) = N(-f_2(k)) - \mathbb P(X_T > k), \qquad \forall k \in \mathbb R,
\end{equation}
which can easily be derived from Black-Scholes formula and the definition of $v(k)$, $\mathrm{Call_{BS}}(k,v(k)) = \esp\bigl[\bigl(\frac {S_T}F - e^k\bigr)^+\bigr]$.
Using the identity $N(-f) = 1 - N(f)$, we also have the equivalent formulation
\begin{equation} \label{vPrime2}
\phi(f_2(k)) v'(k) = -N(f_2(k)) + \mathbb P(X_T \le k), \qquad k \in \mathbb R.
\end{equation}

\begin{remark}
It follows from expression \eqref{vPrime} that $ -1 \leq v'(k)\phi(f_2(k)) \leq 1$, in particular this quantity is bounded, so that the condition $\lim_{k \to \pm \infty} \Psi(k) v'(k) \phi(f_2(k)) = 0$ holds for functions $\Psi$ going to zero at infinity.
\end{remark}

\begin{proof}[Proof of Lemma
\ref{auxiliaryIntegrabilityAndLimits}]
In what follows, $c$ denotes a positive constant that can change from line to line, but does not depend on $k$ nor on any other parameter.
Let $\Psi$ be of exponential growth of order $p$.

$\bullet$ Using the boundedness of $v'(k) \phi(f_2(k))$ from Eq \eqref{vPrime}, we have $|\Psi(k) v'(k) \phi(f_2(k))| \le c \, e^{pk}$. 
Therefore, if $p \maj 0$, $\lim_{k \to -\infty} \Psi(k) v'(k) \phi(f_2(k)) = 0$ and this function is integrable in a neighborhood of $-\infty$.
On the other hand, using again Eq \eqref{vPrime}, the bound $f_2(k) \ge \sqrt{2k}$ for $k > 0$, and the bound on Mill's ratio $N(-f) \le \frac{\phi(-f)}{f}$ for $f \maj 0$, we have
\[
\begin{aligned}
|\Psi(k) v'(k) \phi(f_2(k))| &= |\Psi(k) N(-f_2(k)) - \Psi(k) \Prob(X_T \maj k)|
\\
&\le \frac c{\sqrt{2k}} e^{pk} \phi(\sqrt{2k}) + c \, e^{pk} \Prob(X_T \maj k)
\\
&\le \frac c{\sqrt{2k}} e^{k(p-1)} + c \, e^{pk} \Prob(X_T \maj k),
\qquad \forall k \maj 0.
\end{aligned}
\]
For the second term, note that $\esp[e^{X_T}] = \esp\left[\frac{S_T}F\right] = 1$ entails $\Prob(X_T \maj k) = O(e^{-\alpha k})$ as $k \to \infty$, for every $\alpha \mino 1$ (for a proof of this fact, see \cite[Lemma 4.4]{GulisashSIAM}).
It follows that, if $p \mino 1$, $\lim_{k \to \infty} \Psi(k) v'(k) \phi(f_2(k)) = 0$ and that this function is integrable in a neighborhood of $+\infty$.
Overall, the conclusion of Lemma \ref{auxiliaryIntegrabilityAndLimits} is true for every $p \in (0,1) = (p_-(2), p_+(2))$ and every function $\Psi$ of exponential growth of order $p$ (regardless of the values of $\beta_\pm$).

$\bullet$ According to the first bullet point, we can limit ourselves to $\beta_+ \in[0,2)$.
Assume that $p$ is in the interval \eqref{pBoundsStronger}.
It follows from Eq \eqref{vPrime} and estimate \eqref{f2behavior} that, for every $\beta \in (\beta_+,2)$,
\[
\begin{aligned}
|\Psi(k) v'(k) \phi(f_2(k))| &\le
\frac {|\Psi(k)|}{\sqrt{2k}} \phi \biggl( \biggl( \frac 1{\sqrt{\beta}} + \frac{\sqrt{\beta}}2 \biggr) \sqrt{k} \biggr)
+  |\Psi(k)| \Prob(X_T \maj k)
\\
&= \frac c{\sqrt{2k}} \exp \left(pk  - \frac12 \biggl(\frac 1{\beta} + \frac {\beta}4 + 1 \biggr)k \right)
+  c \, e^{pk} \Prob(X_T \maj k)
\\
&= \frac c{\sqrt{2k}} \exp \left( k (p - p_+ (\beta)) \right)
+  c \, e^{pk} \Prob(X_T \maj k),
\qquad \forall  k > \overline k_\beta.
\end{aligned}
\]
By assumption, $p - p_+(\beta_+) < 0$.
Choosing $\beta$ sufficiently close to $\beta_+$, we have $p - p_+ (\beta) \mino 0$, too (in the particular case $\beta_+ = 0$, we can make $p-p_+(\beta)$ arbitrarily small by taking $\beta > 0$ small enough).

For the second term in the last line, note that the right critical moment $p^*$ satisfies $p^* = \sup \{ \alpha > 0: \mathbb P(X_T > k) = O(e^{-k\alpha}) \mbox{ as } k \to \infty \}$, see again \cite[Lemma 4.4]{GulisashSIAM}.
Therefore, for every $\alpha < p^*$, $e^{pk} \Prob(X_T \maj k) = O(e^{k (p-\alpha)})$ as $k \to \infty$. 
Taking $p< \alpha < p^*$, we can conclude that $\lim_{k \to \infty} \Psi(k) v'(k) \phi(f_2(k)) = 0$ and that $k \mapsto \Psi(k) v'(k) \phi(f_2(k))$ is integrable in a neighborhood of $+\infty$.

$\bullet$ The analogous argument holds for the left side behavior of $\Psi(k) v'(k) \phi(f_2(k))$: let us provide the details for completeness.
From the first bullet point, we can assume $\beta_- \in [0,2)$. By definition, for every $\delta > \beta_-$, there exists $\underline k_\delta < 0$ such that $v(k) < \sqrt{\delta |k|}$ for all $k < \underline k_\delta$.
It follows that
\begin{equation} \label{f2secondEstimate}
\mbox{For every } \beta_- < \delta < 2, \ f_2(k) \le -\left(\frac 1{\sqrt{\delta}} - \frac{\sqrt{\delta}}2 \right) \sqrt{|k|}
=: -r_{\delta} \sqrt{|k|}, \ \ \forall k < \underline k_\delta.
\end{equation}
In order to prove \eqref{f2secondEstimate}, note that for every $a > 1$ and $k < \underline k_\delta$
\[
f_1(k) = -\frac{|k|}{v(k)} - \frac{v(k)}2 = -\left( \frac{|k|}{v(k)} + \frac{a v(k)}2 \right) + \frac{(a-1)v(k)}2 \le -\sqrt{2a|k|} + \frac{(a-1)\sqrt{\delta |k|}}2
\] from which we obtain
$f_1(k) \le -\left( \frac1{\sqrt \delta} + \frac{\sqrt \delta}2 \right) \sqrt{|k|}$
by choosing $a = 2/\delta$.
Then, \eqref{f2secondEstimate} follows from $f_2(k) = f_1(k) + v(k)$.

Consequently, if $p$ is in the interval \eqref{pBoundsStronger}, using Eq \eqref{vPrime2} and estimate \eqref{f2secondEstimate}, for every $\delta \in (\beta_-,2)$ we have
\[
\begin{aligned}
|\Psi(k) v'(k) \phi(f_2(k))| &\le
\frac {|\Psi(k)|}{\sqrt{r_{\delta} |k|}} \phi \left( \Bigl( \frac 1{\sqrt{\delta}} - \frac{\sqrt{\delta}}2 \Bigr) \sqrt{k}\right)
+  |\Psi(k)| \Prob(X_T \mino k)
\\
&= \frac c{\sqrt{r_{\delta} |k|}} \exp \left(pk  - \frac12 \biggl(\frac 1{\delta} + \frac {\delta}4 - 1 \biggr) |k| \right)
+  c \, e^{pk} \Prob(X_T \mino k)
\\
&= \frac c{\sqrt{r_{\delta} |k|}} \exp \left( -|k| (p + p_- (\delta)) \right)
+  c \, e^{pk} \Prob(X_T \mino k),
\qquad \forall  k < \underline k_\delta.
\end{aligned}
\]
By assumption, $p + p_-(\beta_-) > 0$. Choosing $\delta$ sufficiently close to $\beta_-$, we have $p + p_-(\delta) > 0$, too.
For the second term in the last line, we use the property $q^* = \sup \{ \alpha > 0: \mathbb P(X_T \mino k) = O(e^{-|k|\alpha}) \mbox{ as } k \to -\infty \}$ \cite[Lemma 4.4]{GulisashSIAM}, which entails $e^{pk} \Prob(X_T \mino k) = O(e^{-|k|(\alpha-p)})$ for every $\alpha \mino q^*$.
Taking $p \mino \alpha \mino q^*$, we conclude that $\lim_{k \to -\infty} \Psi(k) v'(k) \phi(f_2(k)) = 0$ and that $k \mapsto \Psi(k) v'(k) \phi(f_2(k))$ is integrable at $-\infty$.

Putting the three bullet points together, we have shown that Lemma \ref{auxiliaryIntegrabilityAndLimits} holds for any value of $\beta_\pm \in [0,2]$.
\end{proof}

\subparagraph{Proof of Proposition \ref{p:extensionFukThmWeaker}}

We follow the lines of \cite[Theorems 4.6 and 4.4]{Fukasawa2012}.
Denote $L_{\Psi}(p)$ the LHS of \eqref{extensionFukasawa}: $L_{\Psi}(p) = \int_{-\infty}^{+\infty} [\Psi(g_2(z)) - \Psi'(g_2(z)) + \Psi'(g_1(z)) e^{-g_1(z)}] \phi(z) dz$.
Using the identity $\phi(f_1(k)) = \phi(f_2(k)) e^k$, we have

\begin{equation} \label{identity1}
\begin{aligned}
\int \Psi'(g_1(z)) e^{-g_1(z)} \phi(z) dz &= \int \Psi'(k) e^{-k} \phi(f_1(k)) f_1'(k) dk
\\
&=
\int \Psi'(k) \phi(f_2(k)) (f_2'(k) - v'(k)) dk
\\
&=
\int \Psi'(g_2(z)) \phi(z) dz - \int \Psi'(k) \phi(f_2(k)) v'(k) dk.
\end{aligned}
\end{equation}

Integrating by parts, applying Lemma
\ref{auxiliaryIntegrabilityAndLimits} and using the identity $\phi'(f) = -f \phi(f)$, we get

\begin{equation} \label{identity2}
\begin{aligned}
- \int \Psi'(k) \phi(f_2(k)) v'(k) dk &= -[\Psi(k) \phi(f_2(k)) v'(k))]^{k \to \infty}_{k \to -\infty} 
+ \int \Psi(k) \frac d{dk} [\phi(f_2(k)) v'(k)] dk
\\
&= \int \Psi(k) \left( -f_2(k) f_2'(k) v'(k) + v''(k) \right) \phi(f_2(k)) dk.
\end{aligned}
\end{equation}

It follows from Eqs \eqref{identity1} and \eqref{identity2} that
$L_{\Psi}(p) = \int \Psi(k) \left[ f_2'(k) (1 - f_2(k)v'(k)) + v''(k) \right] \phi(f_2(k)) dk$.
Now recall that, under the assumption that $v$ is twice differentiable, the
density function of $S_T$ at the point $K = F e^k$ is given by
\begin{equation} \label{density}
\begin{aligned}
\frac {d^2 P(K)}{dK^2} = \frac {d^2}{dK^2} \mathbb E [(K-S_T)^+] &= \frac {d^2}{dK^2} \mathrm{P_{BS}} \biggl(K,v\Bigl(\ln \frac K F\Bigr) \biggr)
\\
&= \frac d{dK} \left( N(f_2(k)) +  \partial_{v} \mathrm{P_{BS}} \biggl(K,v\Bigl(\ln \frac K F\Bigr) \biggr) v'(k) \frac 1{K} \right)
\\
&= \frac d{dK} \left( N(f_2(k)) + \phi(f_2(k)) v'(k) \right) 
\\
&= \phi(f_2(k)) \left[f_2'(k) (1-f_2(k) v'(k)) + v''(k) \right] \frac 1{F e^k}.
\end{aligned}
\end{equation}
where we have used the identities
$\partial_{v} \mathrm{P_{BS}}(F e^k,v) = F e^k \phi(f_2(k))$ and, again, $\phi'(f) = -f \phi(f)$.
Applying \eqref{density}, we obtain

\[
L_{\Psi}(p) = \int_{-\infty}^{\infty} \Psi(k) \frac {d^2 P(F e^k)}{dK^2} F e^k dk
= \int_0^\infty \Psi \Bigl(\log \frac K F\Bigr) \frac {d^2 P(K)}{dK^2} dK
= \mathbb E \left[ \Psi \Bigl(\log \frac{S_T} F\Bigr) \right]
\] 
therefore Equation \eqref{extensionFukasawa} is proved under condition \eqref{pBoundsStronger} on $p$.
\qed

\hspace{4mm} We finally have to strengthen Proposition \ref{p:extensionFukThmWeaker} into Theorem \ref{extensionFukThm}. 
If we know that $p^* \ge p_+(\beta_+)$ and $q^* \ge p_-(\beta_-)$, this is immediate.
Recall anyhow that our intention here is \emph{not} to make use of Lee's result \cite{Lee04}, therefore these bounds need to be proved.

The key ingredient will be the following result from the theory of Laplace transforms.
Let $\varphi \in L^1_{loc}(\R_+)$ such that $\varphi \ge 0$ a.e. 
Define $\Phi(p) = \int_0^{\infty} e^{pt} \varphi(t) dt$ and denote $abs(\Phi) = \sup\{p \in \R: \Phi(p) \mino \infty\}$ the abscissa of convergence of $\Phi$, where $\inf \emptyset = -\infty$.
Then, $\Phi$ defines a holomorphic function on $\{\mathrm{Re}(p) \mino abs(\varphi) \}$.

\begin{lemma}[Theorem 2.7.1 in \cite{LaplTransf}] \label{e:lemmaLaplTransform}
Let $\varphi \in L^1_{loc}(\R_+)$ such that $\varphi \ge 0$ a.e. Assume that $-\infty < abs(\varphi) \mino \infty$.
Then, $\Phi$ cannot be extended to a holomorphic function on a neighbourhood of $abs(\varphi)$.
\end{lemma}

Lemma \ref{e:lemmaLaplTransform} allows to prove the bounds we need in order to conclude.

\begin{lemma} \label{l:criticalMomentsFirstInequality}
$p^* \ge p_+(\beta_+)$ and $q^* \ge p_-(\beta_-)$.
\end{lemma}

\begin{proof}
We focus on the first inequality, $p^* \ge p_+(\beta_+)$.
Assume $p^* \mino p_+(\beta_+)$.
We have $M(p) = M_-(p) + M_+(p) := \int_{-\infty}^0 e^{px} f_X(x) dx + \int_0^{\infty} e^{px} f_X(x) dx$, where $f_X \in L^1(\R)$ is the density of $X_T$ (which exists under Assumption \ref{standingAssumptionsV} (i) on the implied volatility $v$).
From Proposition \ref{p:extensionFukThmWeaker}, the identity 
\[
M_+(p) = L(p) - M_-(p)
\]
holds on $\{0 \mino \mathrm{Re}(p) \mino p^* \}$.
By definition of $p^*$, $abs(M_+) = p^*$.
On the other hand, $M_-$ is holomorphic on the half-plane $\{\mathrm{Re}(p) \maj 0 \}$, and $L$ is holomorphic on the strip $\{0 \mino \mathrm{Re}(p) \mino p_+(\beta_+)\}$ by Lemma \ref{auxiliaryIntegrability}.
In other words, the function $L(\cdot) - M_-(\cdot)$ is a holomorphic extension of $M_+$ to the strictly larger strip $\{0 \mino \mathrm{Re}(p) \mino p_+(\beta_+)\}$, contradicting Lemma \ref{e:lemmaLaplTransform}.
The second inequality $q^* \ge p_-(\beta_-)$ is proven analogously.
\end{proof}

As pointed out above, the proof of Theorem \ref{extensionFukThm} is now immediate.

\textbf{Proof of Theorem
\ref{extensionFukThm}}
Lemma \ref{l:criticalMomentsFirstInequality} implies $\min(p^*, p_+(\beta_+)) = p_+(\beta_+)$ and $\min(q^*, p_-(\beta_-)) = p_-(\beta_-)$.
Equation \eqref{extensionFukasawa} then follows from Proposition \ref{p:extensionFukThmWeaker}.
Equation \eqref{moment} is \eqref{extensionFukasawa} for the function
$\Psi(k) = e^{pk}$.

\bibliographystyle{siam}
\bibliography{ReferencesMgfsNormalizedVols}

\begin{thebibliography}{10}

\bibitem{LaplTransf}
{\sc W.~Arendt, C.~Batty, M.~Hieber, and F.~Neubrander}, {\em Vector-valued
  Laplace Transforms and Cauchy Problems}, Monograph in Mathematics 96,
  Springer Basel, second~ed., 2011.

\bibitem{Bergomi2016}
{\sc L.~Bergomi}, {\em Stochastic Volatility Modeling}, Chapman and Hall/CRC,
  2016.

\bibitem{ChrissMorok99}
{\sc N.~Chriss and W.~Morokoff}, {\em Market risk for volatility and variance
  swaps}, Risk, 1 (October 1999), pp.~609--641.

\bibitem{ImpliedVolAtom}
{\sc S.~{De Marco}, C.~Hillairet, and A.~Jacquier}, {\em Shapes of implied
  volatility with positive mass at zero}.
\newblock Forthcoming in SIAM J. Financ. Math.
  {https://arxiv.org/abs/1310.1020}, 2013.

\bibitem{Fukasawa2010}
{\sc M.~Fukasawa}, {\em Normalization for implied volatility}.
\newblock Preprint ArXiv, {https://arxiv.org/abs/1008.5055}, 2010.

\bibitem{Fukasawa2012}
\leavevmode\vrule height 2pt depth -1.6pt width 23pt, {\em The normalizing
  transformation of the implied volatility smile}, Mathematical Finance, 22
  (2012), pp.~753--762.

\bibitem{GathBook}
{\sc J.~Gatheral}, {\em The Volatility Surface: {A} Practitioner's Guide},
  Wiley Finance, 2006.

\bibitem{SSVI}
{\sc J.~Gatheral and A.~Jacquier}, {\em Arbitrage-free svi volatility
  surfaces}, Quantitative Finance, 14 (2014), pp.~59--71.

\bibitem{GulisashSIAM}
{\sc A.~Gulisashvili}, {\em Asymptotic formulas with error estimates for call
  pricing functions and the implied volatility at extreme strikes}, SIAM
  Journal on Financial Mathematics, 1 (2010), pp.~609--641.

\bibitem{Lee04}
{\sc R.~Lee}, {\em The moment formula for implied volatility at extreme
  strikes}, Mathematical Finance, 14 (2004), pp.~469--480.

\bibitem{LeeTransforms04}
\leavevmode\vrule height 2pt depth -1.6pt width 23pt, {\em Option pricing by
  transform methods: Extensions, unification, and error control}, Journal of
  Computational Finance, 7 (2004), pp.~51–--86.

\bibitem{Mat2000}
{\sc A.~Matytsin}, {\em Perturbative analysis of volatility smiles}.
\newblock Presentation at the Columbia Practitioners Conference on the
  Mathematics of Finance, New York, 2000.

\bibitem{RogTeh}
{\sc L.~C.~G. Rogers and M.~R. Tehranchi}, {\em Can the implied volatility
  surface move by parallel shifts?}, Finance and Stochastics, 14 (2010),
  pp.~235--248.

\bibitem{TehPresentation}
{\sc M.~Tehranchi}, {\em No-arbitrage bounds on implied volatility}.
\newblock Conference presentation, PDE\& Finance Stockholm 2007, available at
  {http://www.math.kth.se/pde\_finance07/}.

\end{thebibliography}
    
\end{document}